\documentclass[12pt]{article}

\usepackage{mystyle}
\usepackage[paper=letterpaper,margin=1in]{geometry}
\usepackage{lmodern}
\usepackage{enumerate}
\usepackage[inline,shortlabels]{enumitem}
\usepackage{amsthm}
\usepackage{times}
\usepackage[comma, numbers, sort&compress, sectionbib]{natbib}
\usepackage[displaymath,floats,textmath,graphics,sections]{preview}

\theoremstyle{definition} \newtheorem{lemma}{Lemma}
\theoremstyle{definition} 
\theoremstyle{definition} \newtheorem{theorem}{Theorem}

\newcommand{\bbR}{{\mathbb{R}}}
\newcommand{\bbZ}{{\mathbb{Z}}}

\newcommand{\abs}[1]{|#1|}
\DeclareMathOperator{\poly}{poly}

\newcommand{\calM}{{\mathcal{M}}}
\newcommand{\N}{N}

\newcommand{\nil}{0}
\newcommand{\inil}{\nil}
\newcommand{\jnil}{\nil}

\newcommand{\lexp}[1]{\exp_\lambda(#1)}

\newcommand{\ConfOrFull}[2]{#2}
\newcommand{\QEDWrap}{}

\title{Group Strategyproof Pareto-Stable Marriage with Indifferences
via the Generalized Assignment Game%
\thanks{
  Department of Computer Science,
  University of Texas at Austin,
  2317 Speedway, Stop D9500,
  Austin, Texas 78712--1757.
  Email: \{onur,\,geocklam,\,plaxton\}@cs.utexas.edu.
  This research was supported by NSF Grant CCF--1217980.}
}

\author{Nevzat Onur Domani\c{c}
  \and Chi-Kit Lam
  \and C. Gregory Plaxton
}

\date{October 2017}

\begin{document}

\begin{titlepage}
\maketitle
\thispagestyle{empty}

\begin{abstract}
  We study the variant of the stable marriage problem in which the
  preferences of the agents are allowed to include indifferences.
  We present a mechanism for producing Pareto-stable matchings
  in stable marriage markets with indifferences
  that is group strategyproof for one side of the market.
  Our key technique involves modeling the stable marriage market
  as a generalized assignment game.
  We also show that our mechanism can be implemented efficiently.
  These results can be extended to the college admissions problem
  with indifferences.
\end{abstract}

\end{titlepage}

\section{Introduction}

The stable marriage problem was first introduced by Gale and Shapley~\cite{GS62}.
The stable marriage market involves a set of men and women,
where each agent has ordinal preferences over the agents of
the opposite sex. The goal is to find a disjoint set of
man-woman pairs, called a \emph{matching}, such that no other
man-woman pair prefers each other to their partners in the matching.
Such matchings are said to be \emph{stable}.
When preferences are strict, a unique man-optimal stable
matching exists and can be computed by the man-proposing
deferred acceptance algorithm of Gale and Shapley~\cite{GS62}.
A mechanism is said to be \emph{group strategyproof for the men}
if no coalition of men can be simultaneously matched to strictly
preferred partners by misrepresenting their preferences.
Dubins and Freedman~\cite{DF81} show that the
mechanism that produces man-optimal matchings is group strategyproof
for the men when preferences are strict.
In our work, we focus on group strategyproofness for the men, since no stable mechanism is strategyproof for both men and women~\cite{Rot82}.

We remark that the notion of group strategyproofness used here assumes no
side payments within the coalition of men.
It is known that group strategyproofness
for the men is impossible for the stable marriage problem with strict
preferences when side payments are allowed~\cite[Chap.~4]{roth+s:match}.
This notion of group strategyproofness is also different from
strong group strategyproofness, in which at least one
man in the coalition gets matched to a strictly preferred partner
while the other men in the coalition get matched to weakly preferred
partners. It is known that strong group strategyproofness for the men is
impossible for the stable marriage problem with strict
preferences~\cite[attributed to Gale]{DF81}.

Indifferences in the preferences of agents arise naturally in real-world applications such as school choice~\cite{APR09,EE08,EE15}.
For the marriage problem with indifferences,
Sotomayor~\cite{Sot11} argues that Pareto-stability is an appropriate solution
concept. A matching is said to be \emph{weakly stable} if
no man-woman pair strictly prefers each other to their partners
in the matching. A matching is said to be \emph{Pareto-optimal} if
there is no other matching that is strictly preferred by some agent and
weakly preferred by all agents.
If a matching is both weakly stable and
Pareto-optimal, it is said to be \emph{Pareto-stable}.

Weakly stable matchings, unlike strongly stable
or super-stable matchings~\cite{Irving1994},
always exist. However, not all weakly stable matchings
are Pareto-optimal~\cite{Sot11}.
Pareto-stable matchings can be obtained by applying successive
Pareto-improve\-ments to weakly stable matchings.
Erdil and Ergin~\cite{EE08,EE15} show that this procedure can be carried out
efficiently. Pareto-stable matchings also exist and can be
computed in strongly polynomial time for
many-to-many matchings \cite{Che12} and multi-unit
matchings \cite{CG10}. Instead of using the characterization
of Pareto-improvement chains and cycles, Kamiyama~\cite{Kam14} gives
another efficient algorithm for many-to-many matchings
based on rank-maximal matchings.
However, none of these mechanisms addresses strategyproofness.

We remark that the notion of Pareto-optimality here
is different from \emph{man-Pareto-optimality},
which only takes into account the preferences of the men.
It is known that man-Pareto-optimality is not compatible
with strategyproofness for the stable marriage problem
with indifferences \cite{EE08,Kes10}.
The notion of Pareto-optimality here is also different from
\emph{Pareto-optimality in expected utility}, which permits
Pareto-domination by non-pure outcomes. A result
of Zhou~\cite{Zho90} implies that Pareto-optimality in expected utility
is not compatible with strategyproofness for the stable
marriage problem with indifferences.

Until recently,
it was not known whether a strategyproof Pareto-stable
mechanism exists.
In our recent workshop paper~\cite{DLP17},
we present a generalization of the deferred acceptance
mechanism that is Pareto-stable and strategyproof for the men.
If the market has $n$ agents,
our implementation of this mechanism runs in $O(n^4)$ time,
matching the time bound of the algorithm of 
Erdil and Ergin~\cite{EE08,EE15}\footnote{
The algorithm of Erdil and Ergin
proceeds in two phases.
In the first phase, ties are broken arbitrarily
and the deferred acceptance algorithm is used to
obtain a weakly stable matching.
In the second phase, a sequence of Pareto-improvements
are applied until a Pareto-stable matching is reached.
In App.~A in the full version of~\cite{DLP17}, 
we show that this algorithm does not provide
a strategyproof mechanism.}.
The proof of strategyproofness relies on reasoning about a certain
threshold concept in the stable marriage market,
and this approach seems difficult to extend to
address group strategyproofness.

In this paper, we introduce a new technique useful for
investigating incentive compatibility for coalitions of men.
We present a Pareto-stable mechanism for
the stable marriage problem with indifferences that
is provably group strategyproof for the men,
by modeling the stable marriage market
as an appropriate form of the generalized assignment game.
In Sect.~\ref{sec:implementation} and App.~\ref{app:gener-da-alg}
and~\ref{app:equivalence},
we show that this mechanism coincides with the
generalization of the deferred acceptance mechanism
presented in~\cite{DLP17}. Thus we obtain an
$O(n^4)$-time group strategyproof Pareto-stable mechanism.

\paragraph{The generalized assignment game.}

The assignment game, introduced by Shapley and Shubik~\cite{SS71},
involves a two-sided matching market with monetary transfer
in which agents have unit-slope linear utility functions.
This model has been generalized to allow agents to have continuous, invertible,
and increasing utility functions~\cite{CK81,DG85,Qui84}.
Some models that generalize both the assignment game
and the stable marriage problems have also been developed,
but their models are not concerned with the strategic behavior of agents~\cite{EK00,Sotomayor00}.
The formulation of the generalized assignment game
in this paper follows the presentation of Demange and Gale~\cite{DG85}.

In their paper, Demange and Gale establish various
elegant properties of the generalized assignment game, such as the
lattice property and the existence of one-sided optimal outcomes.
(One-sided optimality or man-optimality is a stronger notion than
one-sided Pareto-optimality or man-Pareto-optimality.)  These
properties are known to hold for the stable marriage market in the
case of strict preferences~\cite[attributed to Conway]{Knuth76}, but fail in the case of weak preferences~\cite[Chap.~2]{roth+s:match}.
Given the similarities between stable marriage markets and
generalized assignment games, it is natural to ask whether stable
marriage markets can be modeled as generalized assignment games.
Demange and Gale discuss this question and state that ``the model
of [Gale and Shapley] is not a special case of our model''.  The basic
obstacle is that it is unclear how to model an agent's preferences
within the framework of a generalized assignment game:
on the one hand, even though ordinal preferences can be converted
into numeric utility values, such preferences are expressed
in a manner that is independent of any monetary transfer;
on the other hand, the framework demands
that there is an amount of money that makes
an agent indifferent between any two agents on the other side
of the market.

In Sect.~\ref{sec:market}, we review key concepts in the work of
Demange and Gale, and introduce the \emph{tiered-slope market} as a
special form of the generalized assignment game in which the slopes of
the utility functions are powers of a large fixed number.  Then, in
Sect.~\ref{sec:marriage}, we describe our approach for converting a
stable marriage market with indifferences into an
associated tiered-slope market.
While these are both two-sided markets that involve the same
set of agents, the utilities achieved under an outcome
in the associated tiered-slope market may not be equal
to the utilities under a corresponding solution in the
stable marriage market.
Nevertheless, we are able to
establish useful relationships between certain sets of solutions to
these two markets.

Our first such result, Theorem~\ref{thm:optimality}, shows that
Pareto-stability in the stable marriage market with indifferences
follows from stability in the associated tiered-slope market, even
though it does not follow from weak stability in the stable marriage
market with indifferences.  This can be seen as a partial analogue to
the case of strict preferences, in which stability in the stable
marriage market implies Pareto-stability~\cite{GS62}.  This also demonstrates
that, in addition to using the deferred acceptance procedure to solve
the generalized assignment game~\cite{CK81}, we can use the
generalized assignment game to solve the stable marriage problem with
indifferences.

In Lemma~\ref{lem:slackness}, we establish that the
utility achieved by any man in a man-optimal solution to the
associated tiered-slope market uniquely determines the tier of
preference to which that man is matched in the stable
marriage market with indifferences.
Another consequence of this lemma is that
any matched man in a man-optimal outcome of the associated
tiered-slope market receives at least one unit of money
from his partner.
We can then deduce that if
a man strictly prefers his partner to a woman,
then the woman has to offer a large amount
of money in order for the man to be indifferent between
her offer and that of his partner.
Since individual rationality prevents
any woman from offering such a large amount of money,
this explains how we overcome the obstacle
of any man being matched with a less preferred woman
in exchange for a sufficiently large payment.

A key result established by Demange and Gale is that the
man-optimal mechanism is group strategyproof for the men.
Using this result and
Lemma~\ref{lem:slackness}, we are able to show in
Theorem~\ref{thm:strategyproofness} that group strategyproofness
for the men in
the stable marriage market with indifferences is achieved by
man-optimality in the associated tiered-slope market, even though it
is incompatible with man-Pareto-optimality in the stable marriage
market with indifferences \cite{EE08,Kes10}.  This can be seen as a
partial analogue to the case of strict preferences, in which
man-optimality implies group strategyproofness \cite{DF81}.

\paragraph{Extending to the college admissions problem.}

We also consider the settings of incomplete preference lists
and one-to-many matchings, in which efficient Pareto-stable
mechanisms are known to exist
\cite{Che12,CG10,EE08,EE15,Kam14}.
Preference lists are incomplete when an agent
declares another agent of the opposite sex to be unacceptable.
Our mechanism is able to support such incomplete preference lists
through an appropriate choice of the reserve utilities of the agents
in the associated tiered-slope market.
In fact, our mechanism also supports indifference between being
unmatched and being matched to some partner.

The one-to-many variant of the stable marriage problem with indifferences
is the college admissions problem with indifferences.
In this model, students and colleges play the roles of men and
women, respectively, and colleges are allowed to be matched with
multiple students, up to their capacities.
We provide the formal definition of the model in App.~\ref{app:college}.
By a simple reduction from college admissions markets to stable marriage markets,
our mechanism is group strategyproof for the students\footnote{A stable
mechanism can be strategyproof only for the side having unit demand,
namely the students~\cite{Rot85}.} and produces
a Pareto-stable matching in polynomial time.

\paragraph{Organization of this paper.}

In Sect.~\ref{sec:market}, we review the generalized assignment game
and define the tiered-slope market.
In Sect.~\ref{sec:marriage}, we introduce the tiered-slope
markets associated with the stable marriage markets with indifferences,
and use them to obtain a group strategyproof, Pareto-stable mechanism.
In Sect.~\ref{sec:implementation} and App.~\ref{app:gener-da-alg}
and~\ref{app:equivalence}, we discuss efficient implementations of the mechanism
and its relationship with the generalization of
the deferred acceptance algorithm presented in~\cite{DLP17}.

\section{Tiered-Slope Market} \label{sec:market}

The generalized assignment game studied by Demange and Gale~\cite{DG85}
involves two disjoint sets $I$ and $J$ of agents, which we call
\emph{men} and \emph{women} respectively. We assume that the sets
$I$ and $J$ do not contain the element $\nil$, which we use to denote
being unmatched.
For each man $i \in I$ and woman $j \in J$, the compensation function
$f_{i, j}(u_i)$ represents the compensation that man $i$ needs to receive
in order to attain utility $u_i$ when he is matched to woman $j$.
Similarly, for each man $i \in I$ and woman $j \in J$, the compensation
function $g_{i, j}(v_j)$ represents the compensation that woman $j$ needs
to receive in order to attain utility $v_j$ when she is matched to man $i$.
Moreover, each man $i \in I$ has a reserve utility $r_i$
and each woman $j \in J$ has a reserve utility $s_j$.

In this paper, we assume that the compensation functions are of the form
\begin{equation*}
f_{i, j}(u_i) = u_i \lambda^{-a_{i, j}}
\qquad \text{and} \qquad
g_{i, j}(v_j) = v_j - (b_{i, j} \N + \pi_i)
\end{equation*}
and the reserve utilities are of the form
\begin{equation*}
r_i = \pi_i \lambda^{a_{i, \jnil}}
\qquad \text{and} \qquad
s_j = b_{\inil, j} \N,
\end{equation*}
where
\begin{equation*}
\pi \in \bbZ^I; \quad
N \in \bbZ; \quad
\lambda \in \bbZ; \quad
a \in \bbZ^{I \times (J \cup \{\jnil\})}; \quad
b \in \bbZ^{(I \cup \{\inil\}) \times J}
\end{equation*}
such that
\begin{equation*}
N > \max_{i \in I} \pi_i \geq \min_{i \in I} \pi_i \geq 1
\end{equation*}
and
\begin{equation*}
\lambda
\geq \max_{(i, j) \in (I \cup \{\inil\}) \times J} (b_{i, j} + 1) \N
\geq \min_{(i, j) \in (I \cup \{\inil\}) \times J} (b_{i, j} + 1) \N
\geq \N.
\end{equation*}
We denote this \emph{tiered-slope market}
as $\calM = (I, J, \pi, N, \lambda, a, b)$.
When $a_{i, j} = 0$ for every man $i \in I$ and
woman $j \in J \cup \{\jnil\}$,
this becomes a \emph{unit-slope market}
$(I, J, \pi, N, \lambda, 0, b)$.
Notice that the compensation functions in a unit-slope market
coincide with those in the assignment game \cite{SS71}
where buyer $j \in J$ has a valuation
of $b_{i, j} \N + \pi_i$ on house $i \in I$.
For better readability, we write $\lexp{\xi}$
to denote $\lambda^\xi$.

A \emph{matching} is a function $\mu \colon I \to J \cup \{0\}$ such that for
any woman $j \in J$, we have $\mu(i) = j$ for at most one man $i \in I$.
Given a matching $\mu$ and a woman $j \in J$, we denote
\begin{equation*}
\mu(j) =
\begin{cases}
i & \text{if $\mu(i) = j$}  \\
0 & \text{if  there is no man $i \in I$ such that $\mu(i) = j$}
\end{cases}
\end{equation*}
An \emph{outcome} is a triple $(\mu, u, v)$, where $\mu$ is a matching,
$u \in \bbR^I$ is the utility vector of the men,
and $v \in \bbR^J$ is the utility vector of the women.
An outcome $(\mu, u, v)$ is \emph{feasible} if the following
conditions hold for every man $i \in I$ and woman $j \in J$.
\begin{enumerate}
\item If $\mu(i) = j$, then $f_{i, j}(u_i) + g_{i, j}(v_j) \leq 0$.
\item If $\mu(i) = \jnil$, then $u_i = r_i$.
\item If $\mu(j) = \inil$, then $v_j = s_j$.
\end{enumerate}
A feasible outcome $(\mu, u, v)$ is \emph{individually rational} if
$u_i \geq r_i$ and $v_j \geq s_j$
for every man $i \in I$ and woman $j \in J$.
An individually rational outcome $(\mu, u, v)$ is \emph{stable} if
$f_{i, j}(u_i) + g_{i, j}(v_j) \geq 0$
for every man $i \in I$ and woman $j \in J$.

A stable outcome $(\mu, u, v)$ is \emph{man-optimal} if for any stable
outcome $(\mu', u', v')$ we have $u_i \geq u'_i$ for every man $i \in I$.
It has been shown that man-optimal outcomes always
exist \cite[Property 2]{DG85}.
Theorem~\ref{thm:manipulation}
below provides a useful group strategyproofness result
for man-optimal outcomes.

\begin{theorem} \label{thm:manipulation}
Let $(\mu, u, v)$ and $(\mu', u', v')$ be man-optimal outcomes of
tiered-slope markets $(I, J, \pi, \allowbreak N, \lambda, a, b)$ and
$(I, J, \pi, N, \lambda, a', b)$, respectively.
If $a \neq a'$, then there exists a man $i_0 \in I$ and
a woman $j_0 \in J \cup \{0\}$ with $a_{i_0, j_0} \neq a'_{i_0, j_0}$
such that 
$u_{i_0} \geq u'_{i_0} \lexp{a_{i_0, \mu'(i_0)} - a'_{i_0, \mu'(i_0)}}$.
\end{theorem}

\begin{proof}
This follows directly from \cite[Theorem 2]{DG85},
which establishes group strategyproofness for the men
in the generalized assignment game with no side payments.
Notice that the value \allowbreak
$u'_{i_0} \lexp{a_{i_0, \mu'(i_0)} - a'_{i_0, \mu'(i_0)}}$
is the true utility of man $i_0$ under matching $\mu'$ as defined
in their paper, both in the case of being matched to
$\mu'(i_0) \neq \jnil$ with compensation
$u'_{i_0} \lexp{-a'_{i_0, \mu'(i_0)}}$
and in the case of being unmatched.
\QEDWrap \end{proof}

\section{Stable Marriage with Indifferences} \label{sec:marriage}

The stable marriage market involves a set $I$ of men
and a set $J$ of women. We assume that the sets
$I$ and $J$ are disjoint and do not contain the element $\nil$,
which we use to denote being unmatched.
The preference relation of each man $i \in I$
is specified by a binary relation $\succeq_i$ over $J \cup \{\jnil\}$
that satisfies transitivity and totality. To allow indifferences,
the preference relation is not required to satisfy anti-symmetry.
Similarly, the preference relation of each woman $j \in J$
is specified by a binary relation $\succeq_j$ over $I \cup \{\inil\}$
that satisfies transitivity and totality.
We denote this \emph{stable marriage market} as $(I, J,
(\succeq_i)_{i \in I}, (\succeq_j)_{j \in J})$.

A \emph{matching} is a function $\mu \colon I \to J \cup \{\jnil\}$ such that
for any woman $j \in J$, we have $\mu(i) = j$ for at most one man $i \in I$.
Given a matching $\mu$ and a woman $j \in J$, we denote
\begin{equation*}
\mu(j) =
\begin{cases}
i & \text{if $\mu(i) = j$}  \\
\inil & \text{if  there is no man $i \in I$ such that $\mu(i) = j$}
\end{cases}
\end{equation*}
A matching $\mu$ is \emph{individually rational} if
$j \succeq_i \jnil$ and $i \succeq_j \inil$
for every man $i \in I$ and woman $j \in J$ such that $\mu(i) = j$.
An individually rational matching $\mu$ is \emph{weakly stable} if
for any man $i \in I$ and woman $j \in J$, either
$\mu(i) \succeq_i j$ or $\mu(j) \succeq_j i$.
(Otherwise, such a man $i$ and woman $j$ form a
\emph{strongly blocking pair}.)

For any matchings $\mu$ and $\mu'$, we say that the binary relation
$\mu \succeq \mu'$ holds if $\mu(i) \succeq_i \mu'(i)$
and $\mu(j) \succeq_j \mu'(j)$ for every man $i \in I$
and woman $j \in J$. A weakly stable matching $\mu$ is \emph{Pareto-stable}
if for any matching $\mu'$ such that $\mu' \succeq \mu$, we have
$\mu \succeq \mu'$. (Otherwise, the matching $\mu$ is
not \emph{Pareto-optimal} because it is \emph{Pareto-dominated}
by the matching $\mu'$.)

A \emph{mechanism} is an algorithm that, given a stable marriage market
$(I, J, (\succeq_i)_{i \in I}, (\succeq_j)_{j \in J})$,
produces a matching $\mu$. A mechanism is said to be
\emph{group strategyproof (for the men)}
if for any two different preference profiles
$(\succeq_i)_{i \in I}$ and $(\succeq'_i)_{i \in I}$,
there exists a man $i_0 \in I$ with preference relation
$\succeq_{i_0}$ different from $\succeq_{i_0}'$ such that
$\mu(i_0) \succeq_{i_0} \mu'(i_0)$, where $\mu$ and $\mu'$
are the matchings produced by the mechanism given $(I, J,
(\succeq_i)_{i \in I}, (\succeq_j)_{j \in J})$ and 
$(I, J, (\succeq'_i)_{i \in I}, \allowbreak (\succeq_j)_{j \in J})$
respectively. (Such a man $i_0$ belongs to the coalition but
is not matched to a strictly preferred woman
by expressing preference relation $\succeq_{i_0}'$
instead of his true preference relation $\succeq_{i_0}$.)

\subsection{The Associated Tiered-Slope Market}

We construct the \emph{tiered-slope market
$\calM = (I, J, \pi, N, \lambda, a, b)$
associated with stable marriage market} $(I, J,
(\succeq_i)_{i \in I}, (\succeq_j)_{j \in J})$ as follows.
We take $N \geq \abs{I} + 1$ and associate with each man
$i \in I$ a fixed and distinct priority
$\pi_i \in \{ 1, 2, \ldots, \abs{I} \}$.
We convert the preference relations
$(\succeq_i)_{i \in I}$ of the men to integer-valued
(non-transferable) utilities
$a \in \bbZ^{I \times (J \cup \{\jnil\})}$ such that
for every man $i \in I$ and
women $j_1, j_2 \in J \cup \{\jnil\}$,
we have
$j_1 \succeq_i j_2$ if and only if
$a_{i, j_1} \geq a_{i, j_2}$.
Similarly, we convert the preference relations
$(\succeq_j)_{j \in J}$ of the women to integer-valued
(non-transferable) utilities
$b \in \bbZ^{(I \cup \{\inil\}) \times J}$ such that
for every woman $j \in J$ and
men $i_1, i_2 \in I \cup \{\inil\}$, we have
$i_1 \succeq_j i_2$ if and only if
$b_{i_1, j} \geq b_{i_2, j} \geq 0$.
Finally, we take
\begin{equation*}
\lambda =
\max_{\substack{i \in I \cup \{\inil\} \\ j \in J}}
(b_{i, j} + 1) \N.
\end{equation*}

In order to achieve group strategyproofness, we require
that $N$ and $\pi$ should not depend on the preferences
$(\succeq_i)_{i \in I}$ of the men.
We further require that $b$ does not depend on
the preferences $(\succeq_i)_{i \in I}$ of the men, and
that $a_{i_0, j_0}$ does not depend on the other preferences
$(\succeq_i)_{i \in I \setminus \{i_0\}}$
for any man $i_0 \in I$ and woman $j_0 \in J \cup \{\jnil\}$.
In other words, a man $i_0 \in I$ is only able to manipulate
his own utilities $(a_{i_0, j})_{j \in J \cup \{\jnil\}}$.
One way to satisfy these conditions is by taking
$a_{i_0, j_0}$ to be the number of women
$j \in J \cup \{\jnil\}$ such that $j_0 \succeq_{i_0} j$
for every man $i_0 \in I$ and woman $j_0 \in J \cup \{\jnil\}$,
and taking $b_{i_0, j_0}$ to be the number of men
$i \in I \cup \{\inil\}$ such that $i_0 \succeq_{j_0} i$
for every man $i_0 \in I \cup \{\inil\}$
and woman $j_0 \in J$.
(These conditions are not used until
Sect.~\ref{sec:strategyproofness}, where we prove group
strategyproofness.)

Intuitively, each woman has a compensation function with the same
form as a buyer in the assignment game \cite{SS71}. The valuation
$b_{i, j} \N + \pi_i$ that woman $j$ assigns to man $i$ has
a first-order dependence on the preferences over the men
and a second-order dependence on the priorities of the men,
which are used to break any ties in her preferences.
From the perspective of man $i$, if he highly prefers a woman $j$,
he assigns a large exponent $a_{i, j}$ in the slope
associated with woman $j$, and thus
expects only a small amount of compensation.

\subsection{Pareto-Stability}

In this subsection, we study the Pareto-stability of matchings
in the stable marriage market that correspond to stable outcomes
in the associated tiered-slope market.
We first show that individual rationality in the associated
tiered-slope market implies individual rationality in the stable
marriage market (Lemmas~\ref{lem:utility} and~\ref{lem:rationality}).
Then, we show that stability in the associated tiered-slope market
implies weak stability in the stable marriage market
(Lemma~\ref{lem:stability}).
Finally, we show that stability in the associated tiered-slope
market is sufficient for Pareto-stability in the stable marriage
market (Lemma~\ref{lem:duality} and Theorem~\ref{thm:optimality}).
The proof of Lemma~\ref{lem:duality} is given in App.~\ref{app:market}.

\begin{lemma} \label{lem:utility}
Let $(\mu, u, v)$ be an individually rational outcome in
tiered-slope market $\calM = (I, J, \pi, N, \allowbreak \lambda, a, b)$.
Let $i \in I$ be a man and $j \in J$ be a woman. Then
\begin{equation*}
0 < \lexp{a_{i, \jnil}} \leq u_i < \lexp{a_{i, \mu(i)} + 1}
\quad \text{ and } \quad
0 \leq b_{\inil, j} \N \leq v_j < (b_{\mu(j), j} + 1) \N.
\end{equation*}
\end{lemma}

\begin{proof}
The lower bounds 
\begin{equation*}
u_i \geq \pi_i \lexp{a_{i, \jnil}}
\geq \lexp{a_{i, \jnil}} > 0
\qquad \text{and} \qquad
v_j \geq b_{\inil, j} \N \geq 0
\end{equation*}
follow directly from individual rationality.
If $\mu(i) = \jnil$, then feasibility implies
\begin{equation*}
u_i = \pi_i \lexp{a_{i, \jnil}} < \lexp{a_{i, \jnil} + 1}.
\end{equation*}
If $\mu(j) = \inil$, then feasibility implies
\begin{equation*}
v_j = b_{\inil, j} \N < (b_{\inil, j} + 1) \N.
\end{equation*}
It remains to show that the upper bounds
hold when $\mu(i) \neq \jnil$ and $\mu(j) \neq \inil$.
Without loss of generality, we may assume that $\mu(i) = j$,
so feasibility implies
\begin{equation*}
u_i \lexp{-a_{i, j}} - (b_{i, j} \N + \pi_i - v_j) \leq 0.
\end{equation*}
Since $u_i \geq 0$ and $v_j \geq 0$, we have
\begin{equation*}
u_i \lexp{-a_{i, j}} - (b_{i, j} \N + \pi_i) \leq 0
\qquad \text{and} \qquad
- (b_{i, j} \N + \pi_i - v_j) \leq 0.
\end{equation*}
Since $\pi_i < \N$ and $b_{i, j} \N + \pi_i < \lambda$, we have
\begin{equation*}
u_i \lexp{-a_{i, j}} - \lambda < 0
\qquad \text{and} \qquad
- (b_{i, j} \N + \N - v_j) < 0.
\end{equation*}
Thus $u_i < \lexp{a_{i, \mu(i)} + 1}$
and $v_j < (b_{\mu(j), j} + 1) \N$.
\QEDWrap \end{proof}

\begin{lemma}[Individual Rationality] \label{lem:rationality}
Let $(\mu, u, v)$ be an individually rational outcome in the
tiered-slope market $\calM = (I, J, \pi, N, \lambda, a, b)$
associated with stable marriage market
$(I, J, \allowbreak {(\succeq_i)_{i \in I}},
\allowbreak {(\succeq_j)_{j \in J}})$.
Then $\mu$ is an individually rational matching in the stable
marriage market.
\end{lemma}

\begin{proof}
Let $i \in I$ be a man and $j \in J$ be a woman.
Then, by Lemma~\ref{lem:utility}, we have
\begin{equation*}
\lexp{a_{i, \jnil}} < \lexp{a_{i, \mu(i)} + 1}
\qquad \text{and} \qquad
b_{\inil, j} \N < (b_{\mu(j), j} + 1) \N.
\end{equation*}
Thus $a_{i, \mu(i)} + 1 > a_{i, \jnil}$ and
$b_{\mu(j), j} + 1 > b_{\inil, j}$,
and hence $a_{i, \mu(i)} \geq a_{i, \jnil}$ and
$b_{\mu(j), j} \geq b_{\inil, j}$.
We conclude that $\mu(i) \succeq_i \jnil$ and $\mu(j) \succeq_j \inil$.
\QEDWrap \end{proof}

\begin{lemma}[Stability] \label{lem:stability}
Let $(\mu, u, v)$ be a stable outcome in the
tiered-slope market $\calM = (I, J, \pi, N, \allowbreak \lambda, a, b)$
associated with stable marriage market
$(I, J, \allowbreak {(\succeq_i)_{i \in I}},
\allowbreak {(\succeq_j)_{j \in J}})$.
Then $\mu$ is a weakly stable matching in the stable marriage market.
\end{lemma}

\begin{proof}
Since the outcome $(\mu, u, v)$ is individually rational
in market $\calM$, Lemma~\ref{lem:rationality} implies that
the matching $\mu$ is individually rational
in the stable marriage market.
It remains to show that there is no strongly blocking pair.

For the sake of contradiction, suppose there exists a man $i \in I$ and
a woman $j \in J$ such that
neither $\mu(i) \succeq_i j$ nor $\mu(j) \succeq_j i$. Then
$a_{i, j} > a_{i, \mu(i)}$ and $b_{i, j} > b_{\mu(j), j}$.
Hence $a_{i, j} \geq a_{i, \mu(i)} + 1$
and $b_{i, j} \geq b_{\mu(j), j} + 1$.
Since $(\mu, u, v)$ is a stable outcome in $\calM$, we have
\begin{align*}
0 \leq{} & u_i \lexp{-a_{i, j}} - (b_{i, j} \N + \pi_i - v_j) \\
<{} & \lexp{a_{i, \mu(i)} + 1} \lexp{-a_{i, j}} -
  (b_{i, j} \N + \pi_i - (b_{\mu(j), j} + 1) \N) \\
\leq{} & 1 - \pi_i,
\end{align*}
where the second inequality follows from Lemma~\ref{lem:utility}.
Thus, $\pi_i < 1$, a contradiction.
\QEDWrap \end{proof}

\begin{lemma} \label{lem:duality}
Let $(\mu, u, v)$ be a stable outcome in the
tiered-slope market $\calM = (I, J, \pi, N, \lambda, \allowbreak
a, b)$. Let $\mu'$ be an arbitrary matching. Then
\begin{equation*}
\sum_{i \in I} \Big( u_i \lexp{-a_{i, \mu'(i)}} - \pi_i \Big)
\geq \sum_{j \in J} \Big( b_{\mu'(j), j} \N - v_j \Big).
\end{equation*}
Furthermore, the inequality is tight if and only if the outcome
$(\mu', u, v)$ is stable.
\end{lemma}

\begin{theorem}[Pareto-stability] \label{thm:optimality}
Let $(\mu, u, v)$ be a stable outcome in the
tiered-slope market $\calM = (I, J, \pi, N, \lambda, a, b)$
associated with stable marriage market
$(I, J, (\succeq_i)_{i \in I}, \allowbreak (\succeq_j)_{j \in J})$.
Then $\mu$ is a Pareto-stable matching in the stable marriage market.
\end{theorem}

\begin{proof}
Since the outcome $(\mu, u, v)$ is stable
in market $\calM$, Lemma~\ref{lem:stability} implies that
the matching $\mu$ is weakly stable in the stable marriage market.
It remains to show that the matching $\mu$ is not Pareto-dominated.

Let $\mu'$ be a matching of the stable marriage market such that
$\mu' \succeq \mu$.
Then $\mu'(i) \succeq_i \mu(i)$ and $\mu'(j) \geq_j \mu(j)$
for every man $i \in I$ and woman $j \in J$.
Hence $a_{i, \mu'(i)} \geq a_{i, \mu(i)}$
and $b_{\mu'(j), j} \geq b_{\mu(j), j}$ for every man $i \in I$
and woman $j \in J$. Since $a_{i, \mu'(i)} \geq a_{i, \mu(i)}$
for every man $i \in I$, we have
\begin{equation*}
\sum_{i \in I} \Big( u_i \lexp{-a_{i, \mu'(i)}} - \pi_i \Big)
\leq \sum_{i \in I} \Big( u_i \lexp{-a_{i, \mu(i)}} - \pi_i \Big).
\end{equation*}
Applying Lemma~\ref{lem:duality} to both sides, we get
\begin{equation*}
\sum_{j \in J} \Big( b_{\mu'(j), j} \N - v_j \Big)
\leq \sum_{j \in J} \Big( b_{\mu(j), j} \N - v_j \Big).
\end{equation*}
Since $b_{\mu'(j), j} \geq b_{\mu(j), j}$ for every woman $j \in J$,
the inequalities are tight.
Hence $a_{i, \mu'(i)} = a_{i, \mu(i)}$
and $b_{\mu'(j), j} = b_{\mu(j), j}$
for every man $i \in I$ and woman $j \in J$.
Thus $\mu(i) \succeq_i \mu'(i)$
and $\mu(j) \succeq_j \mu'(j)$
for every man $i \in I$ and woman $j \in J$.
We conclude that $\mu \succeq \mu'$.
\QEDWrap \end{proof}

\subsection{Group Strategyproofness} \label{sec:strategyproofness}

In this subsection, we study the group strategyproofness of matchings
in the stable marriage market that correspond to man-optimal outcomes
in the associated tiered-slope market.
We first show that the utilities of the men in man-optimal outcomes
in the associated tiered-slope market reflect the utilities of the
men in the stable marriage market (Lemma~\ref{lem:slackness}).
Then we prove group strategyproofness in the stable marriage market
using group strategyproofness in the associated tiered-slope market
(Theorem~\ref{thm:strategyproofness}).

\begin{lemma} \label{lem:slackness}
Let $(\mu, u, v)$ be a man-optimal outcome in the tiered-slope market
$\calM = (I, J, \pi, \allowbreak N, \lambda, a, b)$
associated with stable marriage market
$(I, J, \allowbreak {(\succeq_i)_{i \in I}},
\allowbreak {(\succeq_j)_{j \in J}})$.
Then $\lexp{a_{i, \mu(i)}} \leq u_i < \lexp{a_{i, \mu(i)} + 1}$
for every man $i \in I$.
\end{lemma}

The proof of Lemma~\ref{lem:slackness} is given in App.~\ref{app:market}.
Since the compensation received by a man $i \in I$ matched with
a woman $\mu(i) \neq \jnil$ is given by $u_i \lexp{-a_{i, \mu(i)}}$,
Lemma~\ref{lem:slackness} implies that the amount of compensation
in man-optimal outcomes is at least $1$ and less than $\lambda$.
In fact, no woman is willing to pay $\lambda$ or more
under any individual rational outcome.

\begin{theorem}[Group strategyproofness] \label{thm:strategyproofness}
If a mechanism produces matchings that correspond to man-optimal outcomes
of the tiered-slope markets associated with the stable marriage
markets, then it is group strategyproof and Pareto-stable.
\end{theorem}

\begin{proof}
We have shown Pareto-stability in the stable marriage market
in Theorem~\ref{thm:optimality}.
It remains only to show group strategyproofness.

Let $(I, J, (\succeq_i)_{i \in I}, (\succeq_j)_{j \in J})$
and $(I, J, (\succeq_i')_{i \in I}, (\succeq_j)_{j \in J})$
be stable marriage markets
where ${(\succeq_i)_{i \in I}}$ and $(\succeq_i')_{i \in I}$
are different preference profiles.
Let $(I, J, \allowbreak \pi, N, \lambda, \allowbreak a, b)$ and
$(I, J, \pi, N, \lambda, \allowbreak a', b)$
be the tiered-slope markets associated
with stable marriage markets
$(I, J, \allowbreak
{(\succeq_i)_{i \in I}}, \allowbreak {(\succeq_j)_{j \in J}})$
and $(I, J, \allowbreak
{(\succeq_i')_{i \in I}}, \allowbreak {(\succeq_j)_{j \in J}})$,
respectively. Let $(\mu, u, v)$ and $(\mu', u', v')$
be man-optimal outcomes of the tiered-slope markets
$(I, J, \pi, N, \lambda, a, b)$ and $(I, J, \pi, N, \lambda, a', b)$,
respectively.

Since the preference profiles
$(\succeq_i)_{i \in I}$ and $(\succeq_i')_{i \in I}$
are different, we have $a \neq a'$.
So, by Theorem~\ref{thm:manipulation},
there exists a man $i_0 \in I$ and a woman $j_0 \in J \cup \{\jnil\}$
with $a_{i_0, j_0} \neq a'_{i_0, j_0}$ such that
$u_{i_0} \geq u'_{i_0}
\lexp{a_{i_0, \mu'(i_0)} - a'_{i_0, \mu'(i_0)}}$. Hence
\begin{align*}
\lexp{a_{i_0, \mu(i_0)} + 1}
>{} & u_{i_0} \\
\geq{} & \frac{u'_{i_0}}{\lexp{a'_{i_0, \mu'(i_0)}}} \lexp{a_{i_0, \mu'(i_0)}}  \\
\geq{} & \lexp{a_{i_0, \mu'(i_0)}},
\end{align*}
where the first and third inequalities follow from Lemma~\ref{lem:slackness}.
This shows that $a_{i_0, \mu(i_0)} + 1 > a_{i_0, \mu'(i_0)}$.
Hence $a_{i_0, \mu(i_0)} \geq a_{i_0, \mu'(i_0)}$,
and we conclude that $\mu(i_0) \succeq_{i_0} \mu'(i_0)$.
Also, since $a_{i_0, j_0} \neq a'_{i_0, j_0}$, the preference relations
$\succeq_{i_0}$ and $\succeq_{i_0}'$ are different.
Therefore, the mechanism is group strategyproof.
\QEDWrap \end{proof}

\section{Efficient Implementation} \label{sec:implementation}

The implementation of our group strategyproof Pareto-stable
mechanism for stable marriage with indifferences
amounts to computing a man-optimal outcome
for the associated tiered-slope market.
Since all utility functions in the tiered-slope market
are linear functions, we can perform this computation using the
algorithm of D\"{u}tting et al.~\cite{DHW15}, which was developed for
multi-item auctions.
If we model each woman $j$ as a non-dummy item in the multi-item
auction with price given by utility $v_j$, then the utility
function of each man on each non-dummy item is a linear function
of the price with a negative slope.
Using the algorithm of D\"{u}tting et al.,
we can compute a man-optimal (envy-free) outcome
using $O(n^5)$ arithmetic operations, where $n$ is the total
number of agents. Since $\poly(n)$ precision
is sufficient, our mechanism admits a polynomial-time
implementation.

For the purpose of solving the stable marriage problem,
it is actually sufficient for a mechanism to produce the matching
without the utility vectors $u$ and $v$ of the associated
tiered-slope market.
In App.~\ref{app:gener-da-alg} and~\ref{app:equivalence}, we show that
the generalization of the deferred acceptance algorithm
presented in~\cite{DLP17}
can be used to compute a matching that corresponds to
a man-optimal outcome for the associated tiered-slope markets.
The proof of Theorem~\ref{thm:implementation} is given in
App.~\ref{app:tsm-reveal}.

\begin{theorem} \label{thm:implementation}
There exists an $O(n^4)$-time algorithm that corresponds
to a group strategyproof Pareto-stable mechanism
for the stable marriage market with indifferences,
where $n$ is the total number of men and women.
\end{theorem}

\bibliography{market}

\appendix

\section{Details of the Tiered-Slope Market} \label{app:market}

\begin{proof}[\ConfOrFull{}{Proof of }Lemma~\ref{lem:duality}]
Since $(\mu, u, v)$ is a stable outcome in market $\calM$,
the following conditions hold.
\begin{enumerate}
\item \label{enum:stability}
  $u_i \lexp{-a_{i, \mu'(i)}} \geq
  b_{i, \mu'(i)} \N + \pi_i - v_{\mu'(i)}$
  for every man $i \in I$ such that $\mu'(i) \neq \jnil$.
\item \label{enum:man}
  $u_i \geq \pi_i \lexp{a_{i, \jnil}}$
  for every man $i \in I$ such that $\mu'(i) = \jnil$.
\item \label{enum:woman}
  $v_j \geq b_{\inil, j} \N$ for every woman $j \in J$
  such that $\mu'(j) = \inil$.
\end{enumerate}
Hence, we have
\allowdisplaybreaks
\begin{align*}
& \sum_{i \in I} \Big( u_i \lexp{-a_{i, \mu'(i)}} - \pi_i \Big) \\
={} & \sum_{\substack{i \in I \\ \mu'(i) \neq \jnil}}
  \Big( u_i \lexp{-a_{i, \mu'(i)}} - \pi_i \Big)
  + \sum_{\substack{i \in I \\ \mu'(i) = \jnil}}
  \Big( u_i \lexp{-a_{i, \jnil}} - \pi_i \Big) \\
\geq{} & \sum_{\substack{i \in I \\ \mu'(i) \neq \jnil}}
  \Big( u_i \lexp{-a_{i, \mu'(i)}} - \pi_i \Big) \\
\geq{} & \sum_{\substack{i \in I \\ \mu'(i) \neq \jnil}}
  \Big( b_{i, \mu'(i)} \N - v_{\mu'(i)} \Big) \\
={} & \sum_{\substack{j \in J \\ \mu'(j) \neq \inil}}
  \Big( b_{\mu'(j), j} \N - v_j \Big) \\
={} & \sum_{j \in J} \Big( b_{\mu'(j), j} \N - v_j \Big)
  - \sum_{\substack{j \in J \\ \mu'(j) = \inil}}
  \Big( b_{\inil, j} \N - v_j \Big) \\
\geq{} & \sum_{j \in J} \Big( b_{\mu'(j), j} \N - v_j \Big),
\end{align*}
where the three inequalities follow from conditions
\ref{enum:man}, \ref{enum:stability}, and~\ref{enum:woman}, respectively.

Furthermore, if the outcome $(\mu', u, v)$ is stable, then
conditions \ref{enum:stability}, \ref{enum:man},
and~\ref{enum:woman} are all tight.
Hence, the inequality in the lemma statement is also tight.

Conversely, if the inequality in the lemma statement is tight, then
conditions \ref{enum:stability}, \ref{enum:man}, and~\ref{enum:woman}
are all tight. Hence, the outcome $(\mu', u, v)$ is feasible.
So, the stability of outcome $(\mu', u, v)$ follows
from the stability of outcome $(\mu, u, v)$.
\QEDWrap \end{proof}

We now introduce two lemmas that are useful for proving Lemma~\ref{lem:slackness}:
Lemma~\ref{lem:compatibility} is used to prove Lemma~\ref{lem:slackness-general};
Lemma~\ref{lem:slackness-general} is used to prove both Lemma~\ref{lem:slackness} and Lemma~\ref{lem:man-lower-bound} of App.~\ref{app:tsm-reveal}.

\begin{lemma} \label{lem:compatibility}
  Let $\calM$ be a tiered-slope market $(I, J, \allowbreak \pi, N, \lambda, \allowbreak a, b)$ and let $\calM'$ be a market that is equal to $\calM$ except that the reserve utilities of the men may differ.
  Let $(\mu, u, v)$ be a man-optimal outcome of $\calM'$.
  Let $J_1 \subseteq J$ be a nonempty subset such that $v_{j} \neq s_{j} = b_{0, j} \N$ for every woman $j \in J_1$.
  Then, there exists a man $i_0 \in I$ and a woman $j_1 \in J_1$ such that $\mu(i_0) \notin J_1$ and $f_{i_0, j_1}(u_{i_0}) + g_{i_0, j_1}(v_{j_1}) = u_{i_0} \lexp{-a_{i_0, j_1}} - (b_{i_0, j_1} \N + \pi_{i_0} - v_{j_1}) = 0$.
\end{lemma}

\begin{proof}
This follows directly from \cite[Lemma 4]{DG85},
which shows the existence of a compatible pair $(i_0, j_1)$.
\QEDWrap \end{proof}

\newcommand{\UtilOrResOf}[2]{#1_{#2}}

\newcommand{\MenRes}{r}
\newcommand{\MenResp}{\MenRes'}
\newcommand{\MenRespp}{\MenRes''}

\newcommand{\ManRes}[1]{\UtilOrResOf{\MenRes}{#1}}
\newcommand{\ManResp}[1]{\UtilOrResOf{\MenResp}{#1}}

\newcommand{\CondMenResMinimum}{(i)}
\newcommand{\CondMenResIntegrality}{(ii)}

\begin{lemma} \label{lem:slackness-general}
  Let $\calM = (I, J, \pi, \allowbreak N, \lambda, a, b)$ be the tiered-slope market associated with stable marriage market $(I, J, (\succeq_i)_{i \in I}, (\succeq_j)_{j \in J})$, and let $\MenRes$ denote the reserve utility vector of the men in $\calM$.
  Let $\MenResp$ be a reserve utility vector of the men such that the following conditions hold:
  \CondMenResMinimum\ $\MenResp \geq \MenRes$;
  \CondMenResIntegrality\ for any man $i$ and any integer $k$, either $\ManResp{i} \lexp{k}$ is an integer or $0 < \ManResp{i} \lexp{k} < 1$.
  Let $\calM'$ be the market that is equal to $\calM$ except that the reserve utilities of the men are given by $\MenResp$, and let $(\mu, u, v)$ be a man-optimal outcome in $\calM'$.
  Then $\lexp{a_{i, \mu(i)}} \leq u_i$ for every man $i \in I$.
\end{lemma}

\begin{proof}
  First observe that the individual rationality of $(u, v)$ and condition~\CondMenResMinimum\ imply $u_i \geq \ManResp{i} \geq \ManRes{i} > 0$ for each man $i$.
  The individual rationality of $(u, v)$ also implies $v_j \geq b_{\inil, j} \N \geq 0$ for each woman $j$.
Let $I_0 = \{ i \in I \colon u_i < \lexp{a_{i, \mu(i)}} \}$.
For the sake of contradiction, suppose $I_0$ is nonempty.
Let $J_0 = \{ j \in J \colon j = \mu(i) \text{ for some man } i \in I_0 \}$.
Notice that for every man $i \in I_0$, we have $\mu(i) \neq \jnil$,
for otherwise $u_i = \ManResp{i} \geq \ManRes{i} = \pi_i \lexp{a_{i, \jnil}} \geq \lexp{a_{i, \jnil}}$
by feasibility. Thus $J_0$ is nonempty.

Let $J_1 = \{ j \in J \colon 0 < \gamma \N + \pi_i - v_j < 1
\text{ for some man } i \in I_0 \text{ and } \gamma \in \bbZ
\text{ such that } \gamma \geq 0 \}$.
Notice that for every man $i \in I_0$ and woman $j \in J_0$
such that $j = \mu(i)$, we have
\begin{equation*}
0 < b_{i, j} \N + \pi_i - v_j < 1
\end{equation*}
because individual rationality, condition~\CondMenResMinimum, and the definition of $I_0$ imply that
\begin{align*}
0 < u_i \lexp{-a_{i, j}} < 1 ,
\end{align*}
and feasibility and stability imply that
\begin{align*}
u_i \lexp{-a_{i, j}} = b_{i, j} \N + \pi_i - v_j.
\end{align*}
Thus $J_0 \subseteq J_1$.
Also, for every woman $j \in J_1$,
we have $v_j \neq b_{\inil, j} \N$ by a simple
non-integrality argument.

Therefore, by Lemma~\ref{lem:compatibility}, there exists
a man $i_0 \in I$ and a woman $j_1 \in J_1$ such that
$\mu(i_0) \notin J_1$ and
\begin{equation}
u_{i_0} \lexp{-a_{i_0, j_1}} = b_{i_0, j_1} \N + \pi_{i_0} - v_{j_1}.
\label{eq:tightness}
\end{equation}
Since $j_1 \in J_1$,
there exists $i_1 \in I_0$ and $\gamma_1 \in \bbZ$ such that
$\gamma_1 \geq 0$ and
\begin{equation}
0 < \gamma_1 \N + \pi_{i_1} - v_{j_1} < 1.
\label{eq:modulus}
\end{equation}
Let $j_0 = \mu(i_0)$. We have
\begin{equation}
u_{i_0} =
  \begin{cases}
    (b_{i_0, j_0} \N + \pi_{i_0} - v_{j_0}) \lexp{a_{i_0, j_0}}
    & \text{if } j_0 \neq \jnil \\
    \ManResp{i_0}
    & \text{if } j_0 = \jnil
  \end{cases}
\label{eq:utility}
\end{equation}
since the outcome $(\mu, u, v)$ is stable.
We consider two cases.

Case 1: $j_0 = \jnil$.
Combining (\ref{eq:tightness}), (\ref{eq:modulus}),
and (\ref{eq:utility}), we get
\begin{equation} \label{eq:single}
0 < (\gamma_1 - b_{i_0, j_1}) \N + (\pi_{i_1} - \pi_{i_0})
  + \ManResp{i_0} \lexp{-a_{i_0, j_1}} < 1.
\end{equation}
Let $\Delta$ denote $\ManResp{i_0} \lexp{-a_{i_0, j_1}}$.
By a simple non-integrality argument, we deduce that $\Delta$ is not an integer.
Then, by condition~\CondMenResIntegrality\ of the lemma,
we have $0 < \Delta < 1$. 
Since (\ref{eq:single}) implies that
\begin{equation*}
0 < (\gamma_1 - b_{i_0, j_1}) \N + (\pi_{i_1} - \pi_{i_0})
  + \Delta < 1
\end{equation*}
and since $0 < \Delta < 1$, we have
$\gamma_1 = b_{i_0, j_1}$ and $\pi_{i_1} = \pi_{i_0}$.
Thus $i_1 = i_0$ by the distinctness of $\pi$.
Since $i_0 = i_1 \in I_0$ and $\mu(i) \neq 0$ for every man $i \in I_0$ by our previous remark, we deduce that $\mu(i_0) \neq 0$.
This contradicts $\mu(i_0) = j_0 = 0$.

Case 2: $\mu(i_0) \neq \jnil$. Combining
(\ref{eq:tightness}), (\ref{eq:modulus}), and (\ref{eq:utility}),
we get
\begin{equation}
0 < (\gamma_1 - b_{i_0, j_1}) \N + (\pi_{i_1} - \pi_{i_0})
  + (b_{i_0, j_0} \N + \pi_{i_0} - v_{j_0})
  \lexp{a_{i_0, j_0} - a_{i_0, j_1}} < 1.
\label{eq:match}
\end{equation}
We consider three subcases.

Case 2.1: $a_{i_0, j_0} \leq a_{i_0, j_1} - 1$.
Let $\Delta = (b_{i_0, j_0} \N + \pi_{i_0} - v_{j_0})
\lexp{a_{i_0, j_0} - a_{i_0, j_1}}$.
Notice that
\begin{equation*}
\Delta
\leq (b_{i_0, j_0} \N + \pi_{i_0} - v_{j_0}) \lambda^{-1}
\leq (b_{i_0, j_0} \N + \pi_{i_0} - 0) \lambda^{-1}
< 1,
\end{equation*}
where the second inequality follows from individual rationality and
\begin{equation*}
\Delta = u_{i_0} \lexp{- a_{i_0, j_1}} > 0,
\end{equation*}
where the inequality follows from individual rationality and condition~\CondMenResMinimum.
Since (\ref{eq:match}) implies
\begin{align*}
0 <{} & (\gamma_1 - b_{i_0, j_1}) \N + (\pi_{i_1} - \pi_{i_0}) + \Delta < 1
\end{align*}
and $0 < \Delta < 1$, we have
$\gamma_1 = b_{i_0, j_1}$ and $\pi_{i_1} = \pi_{i_0}$.
Thus $i_1 = i_0$ by the distinctness of $\pi$.
Since $i_0 = i_1 \in I_0$, we have $\mu(i_0) \in J_0
\subseteq J_1$, which is a contradiction.

Case 2.2: $a_{i_0, j_0} = a_{i_0, j_1}$.
Substituting into (\ref{eq:match}), we get
\begin{equation}
0 < (\gamma_1 - b_{i_0, j_1} + b_{i_0, j_0}) \N + \pi_{i_1} - v_{j_0} < 1.
\label{eq:continuation}
\end{equation}
This shows that 
\begin{equation}
\gamma_1 - b_{i_0, j_1} + b_{i_0, j_0} \leq -1,
\label{eq:woman}
\end{equation}
for otherwise $\mu(i_0) = j_0 \in J_1$.
Combining (\ref{eq:woman}) with
the lower bound in (\ref{eq:continuation}), we get
\begin{equation*}
0 < -\N + \pi_{i_1} - v_{j_0} < - v_{j_0},
\end{equation*}
which contradicts the lower bound $v_{j_0} \geq 0$ implied by individual rationality.

Case 2.3: $a_{i_0, j_0} \geq a_{i_0, j_1} + 1$.
The upper bound in (\ref{eq:match}) gives
\begin{align*}
b_{i_0, j_0} \N + \pi_{i_0} - v_{j_0}
<{} & ((b_{i_0, j_1} - \gamma_1) \N + \pi_{i_0} - \pi_{i_1} + 1)
  \lexp{a_{i_0, j_1} - a_{i_0, j_0}} \\
\leq{} & ((b_{i_0, j_1} - 0) \N + \pi_{i_0} - 1 + 1) \lexp{-1} \\
={} & (b_{i_0, j_1}  \N + \pi_{i_0}) \lambda^{-1} \\
<{} & 1.
\end{align*}
This shows that 
\begin{equation}
b_{i_0, j_0} \N + \pi_{i_0} - v_{j_0} \leq 0,
\label{eq:man}
\end{equation}
for otherwise $\mu(i_0) = j_0 \in J_1$.
Combining (\ref{eq:man}) with (\ref{eq:utility}), we get
\begin{equation*}
u_{i_0} \lexp{-a_{i_0, j_0}} \leq 0,
\end{equation*}
which contradicts the lower bound $u_{i_0} > 0$ implied by individual rationality and condition~\CondMenResMinimum.
\QEDWrap \end{proof}

\begin{proof}[\ConfOrFull{}{Proof of }Lemma~\ref{lem:slackness}]
Since the outcome $(\mu, u, v)$ is individually rational in market $\calM$,
Lemma~\ref{lem:utility} implies that $u_i < \lexp{a_{i, \mu(i)} + 1}$
for every man $i \in I$. Lemma~\ref{lem:slackness-general}, when invoked with $\MenResp = \MenRes$, implies the lower bound
$\lexp{a_{i, \mu(i)}} \leq u_i$ for every man $i \in I$.
\QEDWrap \end{proof}

\newcommand{\Id}[1]{\mathit{#1}}

\newcommand{\pr}[1]{\left( #1 \right)}
\newcommand{\cbr}[1]{\left\{ #1 \right\}}
\newcommand{\abr}[1]{\langle #1 \rangle}
\newcommand{\sbr}[1]{\left[ #1 \right]}
\newcommand{\floor}[1]{\left\lfloor #1 \right\rfloor}
\newcommand{\ceil}[1]{\left\lceil #1 \right\rceil}
\newcommand{\mult}{\cdot}

\newcommand{\List}[1]{\abr{#1}}
\newcommand{\Set}[1]{\cbr{#1}}
\newcommand{\SetBuild}[2]{\Set{#1 \mid #2}}
\newcommand{\Seq}[1]{\abr{#1}}
\newcommand{\Tuple}[1]{\pr{#1}}
\newcommand{\Car}[1]{\left\vert #1 \right\vert}
\newcommand{\Abs}[1]{\left\vert #1 \right\vert}

\section{A Generalization of the Deferred Acceptance Algorithm}
\label{app:gener-da-alg}

This appendix presents a generalization of the deferred acceptance algorithm that provides a strategyproof and Pareto-stable mechanism for the stable marriage problem with indifferences.
The algorithm admits an $O(n^4)$-time implementation, where $n$ denotes the number of agents in the market.
We only give the definitions and lemmas that are required by App.~\ref{app:equivalence}, where we show that the mechanism presented in this appendix coincides with the group strategyproof mechanism of Sect.~\ref{sec:marriage};
full version of~\cite{DLP17} includes the details and omitted proofs.

The assignment game of Shapley and Shubik~\cite{SS71} can be viewed as an auction with multiple distinct items where each bidder is seeking to acquire at most one item.
This class of \emph{unit-demand auctions} has been heavily studied in the literature (see, e.g., Roth and Sotomayor~\cite[Chapter~8]{roth+s:match}).
In App.~\ref{app:uap}, we define the notion of a ``unit-demand auction with priorities'' (UAP), which extends the notion of a unit-demand auction, and we establish a number of useful properties of UAPs.
Appendix~\ref{app:iuap} builds on the UAP notion to define the notion of an ``iterated UAP'' (IUAP), and defines a mapping from an IUAP to a UAP by describing an algorithm that generalizes the deferred algorithm.
App~\ref{app:alg-smiw} presents our polynomial-time algorithm that provides a strategyproof
Pareto-stable mechanism.

\subsection{Unit-Demand Auctions with Priorities}
\label{app:uap}

In this appendix, we first formally define the notion of a unit-demand auction with priorities (UAP).
Then, we describe an associated matroid for a given UAP and we use this matroid to define the notion of a ``greedy MWM''.
We start with some useful definitions.

\newcommand{\Reals}{\mathbb{R}}

\newcommand{\Item}{v}
\newcommand{\Itemp}{\Item'}
\newcommand{\Itempp}{\Item''}
\newcommand{\ItemSub}[1]{\Item_{#1}}

\newcommand{\Items}{V}
\newcommand{\Itemsp}{\Items'}

\newcommand{\UBid}{\beta}

\newcommand{\Offer}{x}
\newcommand{\Offerp}{\Offer'}
\newcommand{\OfferSub}[1]{\Offer_{#1}}
\newcommand{\OfferOpt}{\Opt{\Offer}}

A \emph{(unit-demand) bid $\UBid$ for a set of items $\Items$} is a subset of $\Items \times \Reals$ such that no two pairs in $\UBid$ share the same first component.
(So $\UBid$ may be viewed as a partial function from $\Items$ to $\Reals$.)

\newcommand{\Opt}[1]{{#1}^*}

\newcommand{\Bidder}{u}
\newcommand{\Bidderp}{\Bidder'}
\newcommand{\Bidderpp}{\Bidder''}
\newcommand{\BidderSub}[1]{\Bidder_{#1}}
\newcommand{\BidderOpt}{\Opt{\Bidder}}

\newcommand{\BidId}{\alpha} 
\newcommand{\Pri}{z} 
\newcommand{\Prip}{\Pri'}
\newcommand{\Pripp}{\Pri''}
\newcommand{\BidIdSub}[1]{\BidId_{#1}}
\newcommand{\PriSub}[1]{\Pri_{#1}}
\newcommand{\PriOpt}{\Opt{\Pri}}

\newcommand{\BId}[1]{\Id{id}(#1)}
\newcommand{\BUBid}[1]{\Id{bid}(#1)}
\newcommand{\BPriId}{\Id{priority}}
\newcommand{\BPri}[1]{\BPriId(#1)}

\newcommand{\BItemsId}{\Id{items}}
\newcommand{\BItems}[1]{\BItemsId(#1)}

A \emph{bidder $\Bidder$ for a set of items $\Items$} is a triple $(\BidId, \UBid, \Pri)$ where $\BidId$ is an integer ID, $\UBid$ is a bid for $\Items$, and $\Pri$ is a real priority.
For any bidder $\Bidder = (\BidId, \UBid, \Pri)$, we define $\BId{\Bidder}$ as $\BidId$, $\BUBid{\Bidder}$ as $\UBid$, $\BPri{\Bidder}$ as $\Pri$, and $\BItems{\Bidder}$ as the union, over all $(\Item, \Offer)$ in $\UBid$, of $\Set{\Item}$.

\newcommand{\Bidders}{U}
\newcommand{\Biddersp}{\Bidders'}
\newcommand{\Bidderspp}{\Bidders''}
\newcommand{\BiddersSub}[1]{\Bidders_{#1}}

\newcommand{\Auc}{A}
\newcommand{\Aucp}{\Auc'}
\newcommand{\Aucpp}{\Auc''}
\newcommand{\AucSub}[1]{\Auc_{#1}}

A \emph{unit-demand auction with priorities (UAP)} is a pair $\Auc = (\Bidders, \Items)$ satisfying the following conditions:
$\Items$ is a set of items;
$\Bidders$ is a set of bidders for $\Items$;
each bidder in $\Bidders$ has a distinct ID\@.

\newcommand{\Edge}{e}
\newcommand{\Edges}{E}

\newcommand{\WId}{\Id{w}}
\newcommand{\WEdge}[1]{\WId(#1)}
\newcommand{\WEdges}[1]{\WId(#1)}
\newcommand{\WPair}[2]{\WId(#1, #2)}
\newcommand{\WMatch}[1]{\WEdges{#1}}

A UAP $\Auc = (\Bidders, \Items)$ may be viewed as an edge-weighted bipartite graph, where the set of edges incident on bidder $\Bidder$ correspond to $\BUBid{\Bidder}$:
for each pair $(\Item, \Offer)$ in $\BUBid{\Bidder}$, there is an edge $(\Bidder, \Item)$ of weight $\Offer$.
We refer to a matching (resp., maximum-weight matching (MWM), maximum-cardinality MWM (MCMWM)) in the associated edge-weighted bipartite graph as a matching (resp., MWM, MCMWM) of $\Auc$.
For any edge $\Edge = (\Bidder, \Item)$ in a given UAP, the associated weight is denoted $\WEdge{\Edge}$ or $\WPair{\Bidder}{\Item}$.
For any set of edges $\Edges$, we define $\WEdges{\Edges}$ as $\sum_{\Edge \in \Edges} \WEdge{\Edge}$.

\newcommand{\Match}{M}
\newcommand{\Matchp}{\Match'}
\newcommand{\Matchpp}{\Match''}
\newcommand{\Matchppp}{\Match'''}
\newcommand{\MatchSub}[1]{\Match_{#1}}
\newcommand{\MatchSubp}[1]{\Match_{#1}'}
\newcommand{\MatchOpt}{\Opt{\Match}}

\newcommand{\Indeps}{\mathcal{I}}

\begin{lemma}
  \label{lem:matroid}
  Let $\Auc = (\Bidders, \Items)$ be a UAP, and let $\Indeps$ denote the set of all subsets $\Biddersp$ of $\Bidders$ such that there exists an MWM of $\Auc$ that matches every bidder in $\Biddersp$.
  Then $(\Bidders, \Indeps)$ is a matroid.
\end{lemma}

\newcommand{\Path}{P}
\newcommand{\PathOther}{Q}
\newcommand{\Cycle}{C}
\newcommand{\PathCycle}{Q}
\newcommand{\PathCycleSub}[1]{\PathCycle_{#1}}
\newcommand{\PathSet}{\mathcal{\Path}}
\newcommand{\PathCycleColl}{\mathcal{S}}

\newcommand{\Matroid}[1]{\Id{matroid(#1)}}

For any UAP $\Auc$, we define $\Matroid{\Auc}$ as the matroid of Lemma~\ref{lem:matroid}.

For any UAP $\Auc = (\Bidders, \Items)$ and any independent set $\Biddersp$ of $\Matroid{\Auc}$, we define the \emph{priority of $\Biddersp$} as the sum, over all bidders $\Bidder$ in $\Biddersp$, of $\BPri{\Bidder}$.
For any UAP $\Auc$, the matroid greedy algorithm can be used to compute a maximum-priority maximal independent set of $\Matroid{\Auc}$.

\newcommand{\MatchedBidders}[1]{\Id{matched}(#1)}

For any matching $\Match$ of a UAP $\Auc = (\Bidders, \Items)$, we define $\MatchedBidders{\Match}$ as the set of all bidders in $\Bidders$ that are matched in $\Match$.
We say that an MWM $\Match$ of a UAP $\Auc$ is \emph{greedy} if $\MatchedBidders{\Match}$ is a maximum-priority maximal independent set of $\Matroid{\Auc}$.

\newcommand{\MPriority}[1]{\Id{priority}(#1)}

For any matching $\Match$ of a UAP, we define the \emph{priority of $\Match$}, denoted $\MPriority{\Match}$, as the sum, over all bidders $\Bidder$ in $\MatchedBidders{\Match}$, of $\BPri{\Bidder}$.
Thus an MWM is greedy if and only if it is a maximum-priority MCMWM\@.

\begin{lemma}
  \label{lem:uap-distribution}
  All greedy MWMs of a given UAP have the same distribution of priorities.
\end{lemma}

\newcommand{\Greedy}[2]{\Id{greedy}(#1, #2)}

For any UAP $\Auc$ and any real priority $\Pri$, we define $\Greedy{\Auc}{\Pri}$ as the (uniquely defined, by Lemma~\ref{lem:uap-distribution}) number of matched bidders with priority $\Pri$ in any greedy MWM of $\Auc$.

\subsection{Iterated Unit-Demand Auctions with Priorities}
\label{app:iuap}

In this appendix, we first formally define the notion of an iterated unit-demand auction with priorities (IUAP).
An IUAP allows the bidders, called ``multibidders'' in this context, to have a sequence of unit-demand bids instead of a single unit-demand bid.
Then we define a mapping from an IUAP to a UAP by describing an algorithm that generalizes the deferred algorithm.
We start with some useful definitions.

\newcommand{\MBidder}{t}
\newcommand{\MBidderp}{\MBidder'}
\newcommand{\MBidderSub}[1]{\MBidder_{#1}}

\newcommand{\MBidders}{T}
\newcommand{\MBiddersp}{\MBidders'}

\newcommand{\BidderSeq}{\sigma}
\newcommand{\BidderSeqp}{\BidderSeq'}
\newcommand{\BidderSeqSub}[1]{\BidderSeq_{#1}}

\newcommand{\MBPri}[1]{\Id{priority}(#1)}
\newcommand{\MBBidder}[2]{\Id{bidder}(#1, #2)}
\newcommand{\MBBidders}[2]{\Id{bidders}(#1, #2)}
\newcommand{\MBBiddersAll}[1]{\Id{bidders}(#1)}

A \emph{multibidder $\MBidder$ for a set of items $\Items$} is a pair $(\BidderSeq, \Pri)$ where $\Pri$ is a real priority and $\BidderSeq$ is a sequence of bidders for $\Items$ such that all the bidders in $\BidderSeq$ have distinct IDs and a common priority $\Pri$.
We define $\MBPri{\MBidder}$ as $\Pri$.
For any integer $i$ such that $1 \leq i \leq \Car{\BidderSeq}$, we define $\MBBidder{\MBidder}{i}$ as the bidder $\BidderSeq(i)$.
For any integer $i$ such that $0 \leq i \leq \Car{\BidderSeq}$, we define $\MBBidders{\MBidder}{i}$ as $\SetBuild{\MBBidder{\MBidder}{j}}{1 \leq
j \leq i}$.
We define $\MBBiddersAll{\MBidder}$ as $\MBBidders{\MBidder}{\Car{\BidderSeq}}$.

\newcommand{\IAuc}{B}
\newcommand{\IAucp}{\IAuc'}
\newcommand{\IAucpp}{\IAuc''}
\newcommand{\IAucppp}{\IAuc'''}
\newcommand{\IAucSub}[1]{\IAuc_{#1}}

\newcommand{\IABidders}[1]{\Id{bidders}(#1)}

An \emph{iterated UAP (IUAP)} is a pair $\IAuc = (\MBidders, \Items)$ where $\Items$ is a set of items and $\MBidders$ is a set of multibidders for $\Items$.
In addition, for any distinct multibidders $\MBidder$ and $\MBidderp$ in $\MBidders$, the following conditions hold:
$\MBPri{\MBidder} \not= \MBPri{\MBidderp}$;
if $\Bidder$ belongs to $\MBBiddersAll{\MBidder}$ and $\Bidderp$ belongs to $\MBBiddersAll{\MBidderp}$, then $\BId{\Bidder} \not = \BId{\Bidderp}$.
For any IUAP $\IAuc = (\MBidders, \Items)$, we define $\IABidders{\IAuc}$ as the union, over all $\MBidder$ in $\MBidders$, of $\MBBiddersAll{\MBidder}$.

\newcommand{\AlgToUAP}{\textsc{ToUap}}

Having defined the notion of an IUAP, we now describe an algorithm \AlgToUAP\ that maps a given IUAP to a UAP.
Algorithm \AlgToUAP\ generalizes the deferred acceptance algorithm.
In each iteration of the deferred acceptance algorithm, an arbitrary single man is chosen, and this man reveals his next choice.
In each iteration of \AlgToUAP, an arbitrary single multibidder is chosen, and this multibidder reveals its next bid.
We state in Lemma~\ref{lem:iuap-confluence} that, like the deferred acceptance algorithm, algorithm \AlgToUAP\ is confluent:
the output does not depend on the nondeterministic choices made during an execution.

\newcommand{\Prefix}[2]{\Id{prefix}(#1, #2)}

Let $\Auc$ be a UAP $(\Bidders, \Items)$ and let $\IAuc$ be an IUAP $(\MBidders, \Items)$.
The predicate $\Prefix{\Auc}{\IAuc}$ is said to hold if $\Bidders \subseteq \IABidders{\IAuc}$ and for any multibidder $\MBidder$ in $\MBidders$, $\Bidders \cap \MBBiddersAll{\MBidder} = \MBBidders{\MBidder}{i}$ for some $i$.

\newcommand{\Conf}{C}
\newcommand{\Confp}{\Conf'}
\newcommand{\Confpp}{\Conf''}
\newcommand{\ConfSub}[1]{\Conf_{#1}}

A \emph{configuration $\Conf$} is a pair $(\Auc, \IAuc)$ where $\Auc$ is a UAP, $\IAuc$ is an IUAP, and $\Prefix{\Auc}{\IAuc}$ holds.

\newcommand{\MBidderOfB}[2]{\Id{multibidder}(#1, #2)}

Let $\Conf = (\Auc, \IAuc)$ be a configuration, where $\Auc = (\Bidders, \Items)$ and $\IAuc = (\MBidders, \Items)$, and let $\Bidder$ be a bidder in $\Bidders$.
Then we define $\MBidderOfB{\Conf}{\Bidder}$ as the unique multibidder $\MBidder$ in $\MBidders$ such that $\Bidder$ belongs to $\MBBiddersAll{\MBidder}$.

\newcommand{\BiddersOfMB}[2]{\Id{bidders}(#1, #2)}

Let $\Conf = (\Auc, \IAuc)$ be a configuration where $\Auc = (\Bidders, \Items)$ and $\IAuc = (\MBidders, \Items)$.
For any $\MBidder$ in $\MBidders$, we define $\BiddersOfMB{\Conf}{\MBidder}$ as $\SetBuild{\Bidder \in \Bidders}{\MBidderOfB{\Conf}{\Bidder} = \MBidder}$.

\newcommand{\CReady}[1]{\Id{ready}(#1)}

Let $\Conf = (\Auc, \IAuc)$ be a configuration where $\IAuc = (\MBidders, \Items)$.
We define $\CReady{\Conf}$ as the set of all bidders $\Bidder$ in $\IABidders{\IAuc}$ such that $\Greedy{\Auc}{\BPri{\Bidder}} = 0$ and $\Bidder = \MBBidder{\MBidder}{\Car{\BiddersOfMB{\Conf}{\MBidder}} + 1}$ where $\MBidder = \MBidderOfB{\Conf}{\Bidder}$.

\begin{algorithm}
\caption{\AlgToUAP$(\IAuc)$}
\label{alg:to-uap}
\begin{algorithmic}[1]
  \Require An IUAP $\IAuc = (\MBidders, \Items)$
  \State $\Auc \gets (\emptyset, \Items)$
  \State $\Conf \gets (\Auc, \IAuc)$
  \While{$\CReady{\Conf}$ is nonempty}
    \State $\Auc \gets \Auc +$ an arbitrary bidder in $\CReady{\Conf}$\label{line:reveal}
    \State $\Conf \gets (\Auc, \IAuc)$
  \EndWhile
  \State \Return $\Auc$
\end{algorithmic}
\end{algorithm}

Our algorithm for mapping an IUAP to a UAP is Algorithm~\ref{alg:to-uap}.
The input is an IUAP $\IAuc$ and the output is a UAP $\Auc$ such that $\Prefix{\Auc}{\IAuc}$ holds.
The algorithm starts with the UAP consisting of all the items in $\Items$ but no bidders.
At this point, no bidder of any multibidder is ``revealed''.
Then, the algorithm iteratively and chooses an arbitrary ``ready'' bidder and ``reveals'' it by adding it to the UAP that is maintained in the program variable $\Auc$.
A bidder $\Bidder$ associated with some multibidder $\MBidder = (\BidderSeq, \Pri)$ is ready if $\Bidder$ is not revealed and for each bidder $\Bidderp$ that precedes $\Bidder$ in $\BidderSeq$, $\Bidderp$ is revealed and is not matched in any greedy MWM of $\Auc$.
It is easy to verify that the predicate $\Prefix{\Auc}{\IAuc}$ is an invariant of the algorithm loop:
if a bidder $\Bidder$ belonging to a multibidder $\MBidder$ is to be revealed at an iteration, and $\Bidders \cap \MBBiddersAll{\MBidder} = \MBBidders{\MBidder}{i}$ for some integer $i$ at the beginning of this iteration, then $\Bidders \cap \MBBiddersAll{\MBidder} = \MBBidders{\MBidder}{i+1}$ after revealing $\Bidder$, where $(\Bidders, \Items)$ is the UAP that is maintained by the program variable $\Auc$ at the beginning of the iteration.
No bidder can be revealed more than once since a bidder cannot be ready after it has been revealed;
it follows that the algorithm terminates.
We now give some useful definitions and lemmas, and we state that the output of the algorithm is uniquely determined (Lemma~\ref{lem:iuap-confluence}), even though the bidder that is revealed in each iteration is chosen nondeterministically.

\newcommand{\CTailId}{\Id{tail}}
\newcommand{\CTail}[1]{\CTailId(#1)}

For any configuration $\Conf = (\Auc, \IAuc)$, we define the predicate $\CTail{\Conf}$ to hold if for any bidder $\Bidder$ that is matched in some greedy MWM of $\Auc$, we have $\Bidder = \MBBidder{\MBidder}{\Car{\BiddersOfMB{\Conf}{\MBidder}}}$ where $\MBidder$ denotes $\MBidderOfB{\Conf}{\Bidder}$.

\begin{lemma}
  \label{lem:iuap-one}
  Let $\Conf = (\Auc, \IAuc)$ be a configuration where $\IAuc = (\MBidders, \Items)$ and assume that $\CTail{\Conf}$ holds.
  Then $\Greedy{\Auc}{\MBPri{\MBidder}} \leq 1$ for each $\MBidder$ in $\MBidders$.
\end{lemma}

\begin{lemma}
  \label{lem:iuap-tail}
  The predicate $\CTail{\Conf}$ is an invariant of the Alg.~\ref{alg:to-uap} loop.
\end{lemma}

\begin{lemma}
  \label{lem:iuap-confluence}
  Let $\IAuc = (\MBidders, \Items)$ be an IUAP\@.
  Then all executions of Alg.~\ref{alg:to-uap} on input $\IAuc$ produce the same output.
\end{lemma}

\newcommand{\UAPId}{\Id{uap}}
\newcommand{\UAP}[1]{\UAPId(#1)}

For any IUAP $\IAuc$, we define $\UAP{\IAuc}$ as the unique (by Lemma~\ref{lem:iuap-confluence}) UAP returned by any execution of Alg.~\ref{alg:to-uap} on input $\IAuc$.

\subsection{The Algorithm}
\label{app:alg-smiw}

\newcommand{\Men}{I}
\newcommand{\Women}{J}

\newcommand{\Man}{i}
\newcommand{\Woman}{j}
\newcommand{\Manp}{\Man'}
\newcommand{\Womanp}{\Woman'}
\newcommand{\Womanpp}{\Woman''}

\newcommand{\ManSub}[1]{\Man_{#1}}
\newcommand{\WomanSub}[1]{\Woman_{#1}}

\newcommand{\Unmatched}{\nil}

\newcommand{\EachMan}{\Man \in \Men}
\newcommand{\EachWoman}{\Woman \in \Women}

\newcommand{\ManOrdPref}[1]{\succeq_{#1}}
\newcommand{\WomanOrdPref}[1]{\succeq_{#1}}
\newcommand{\ManOrdPrefS}{\ManOrdPref{\Man}}
\newcommand{\WomanOrdPrefS}{\WomanOrdPref{\Woman}}

\newcommand{\MarMatch}{\mu}
\newcommand{\MarMatchp}{\MarMatch'}
\newcommand{\MarMatchOpt}{\MarMatch^*}
\newcommand{\MarMatchSub}[1]{\MarMatch_{#1}}

\newcommand{\MOf}[2]{#1(#2)}

\newcommand{\toMBidder}[1]{\Id{multibidder}(#1)}
\newcommand{\toItem}[1]{\Id{item}(#1)}
\newcommand{\toDummy}[1]{\Id{item}_{\Unmatched}(#1)}

The computation of a matching for an instance of the stable marriage market with indifferences is shown in Alg.~\ref{alg:smiw}.
In order to suit the presentation of the current paper, Alg.~\ref{alg:smiw} is expressed using a different notation than in~\cite{DLP17}.
For each woman $\Woman$, we construct an item, denoted $\toItem{\Woman}$, in line~\ref{line:smiw-item}.
For each man $\Man$, we construct a dummy item, denoted $\toDummy{\Man}$, in line~\ref{line:smiw-dummy}, and a multibidder, denoted $\toMBidder{\Man}$, in line~\ref{line:smiw-bidder}, by examining the tiers of preference of the men and the utilities of the women.
The set $\SetBuild{\PriSub{\Man}}{\EachMan}$ of priorities of the multibidders is equal to $\Set{1, \dotsc, \Car{\Men}}$, and we assume that the men have no control over the assignment of the priorities.
These multibidders and items form an IUAP, from which we obtain a UAP and a greedy MWM $\Match$.
Finally, in line~\ref{line:smiw-return}, we use $\Match$ to determine the match of each man in the solution to the stable marriage instance.

\newcommand{\WomenUtilFun}{\psi}
\newcommand{\WomanUtilFunName}[1]{\WomenUtilFun_{#1}}
\newcommand{\WomanUtilFun}[2]{\WomanUtilFunName{#1}(#2)}

\newcommand{\Tier}[2]{\tau_{#1}(#2)}
\newcommand{\NTiers}[1]{K_{#1}}

\begin{algorithm}[htb]
\caption{}
\label{alg:smiw}
\begin{algorithmic}[1]
  \ForAll{$\EachWoman$}
    \State Convert the preference relation $\WomanOrdPrefS$ of woman $\Woman$ into utility function $\WomanUtilFunName{\Woman} \colon \Men + \Unmatched \to \Reals$ that satisfies the following conditions:
    $\WomanUtilFun{\Woman}{\Unmatched} = 0$;
    for any $\Man$ and $\Manp$ in $\Men + \Unmatched$, we have $\Man \WomanOrdPrefS \Manp$ if and only if $\WomanUtilFun{\Woman}{\Man} \geq \WomanUtilFun{\Woman}{\Manp}$.
    This utility assignment should not depend on the preferences of the men.
    \State Construct an item, denoted $\toItem{\Woman}$, corresponding to woman $\Woman$.\label{line:smiw-item}
  \EndFor
  \ForAll{$\EachMan$}
    \State Partition the set $\Women + \Unmatched$ into tiers $\Tier{\Man}{1}, \ldots, \Tier{\Man}{\NTiers{\Man}}$ according to the preference relation of man $\Man$, such that for any $\Woman$ in $\Tier{\Man}{k}$ and $\Womanp$ in $\Tier{\Man}{k'}$, we have $\Woman \ManOrdPrefS \Womanp$ if and only if $k \leq k'$.
    \State Construct a dummy item, denoted $\toDummy{\Man}$, corresponding to man $\Man$.\label{line:smiw-dummy}
    \State Construct a multibidder $(\BidderSeqSub{\Man}, \PriSub{\Man})$, denoted $\toMBidder{\Man}$, corresponding to man $\Man$.
    The priority $\PriSub{\Man}$ is uniquely chosen from the set $\Set{1, \dotsc, \Car{\Men}}$.
    The sequence $\BidderSeqSub{\Man}$ has $\NTiers{\Man}$ bidders such that for each bidder $\BidderSeqSub{\Man}(k)$, we define $\BItems{\BidderSeqSub{\Man}(k)}$ as $\SetBuild{\toItem{\Woman}}{\Woman \in \Tier{\Man}{k}}$ and $\WPair{\BidderSeqSub{\Man}(k)}{\toItem{\Woman}}$ as $\WomanUtilFun{\Woman}{\Man}$, where $\toItem{\Unmatched}$ denotes $\toDummy{\Man}$, and $\WomanUtilFun{\Unmatched}{\Man}$ denotes $0$.\label{line:smiw-bidder}
  \EndFor
  \State $\IAuc = (\MBidders, \Items) = (\SetBuild{\toMBidder{\Man}}{\EachMan}, \SetBuild{\toItem{\Woman}}{\EachWoman} \cup \SetBuild{\toDummy{\Man}}{\EachMan} )$.\label{line:smiw-iuap}
  \State $\Auc = \UAP{\IAuc}$.\label{line:smiw-uap}
  \State Compute a greedy MWM $\Match$ of UAP $\Auc$.\label{line:smiw-greedy}
  \State Output matching $\MarMatch$ such that for each man $\Man$ in $\Men$ and each woman $\Woman$ in $\Women$, we have $\MarMatch(\Man) = \Woman$ if and only if $\BidderSeqSub{\Man}(k)$ is matched to item $\toItem{\Woman}$ in $\Match$ for some $k$.\label{line:smiw-return}
\end{algorithmic}
\end{algorithm}

Algorithm~\ref{alg:smiw} implements a strategyproof Pareto-stable mechanism for the stable marriage problem with indifferences~\cite[Theorem~1]{DLP17}.
The algorithm admits an $O(n^4)$-time implementation, since lines~\ref{line:smiw-uap} and~\ref{line:smiw-greedy} can be implemented in $O(n^4)$ time using the version of the incremental Hungarian method discussed in~\cite[Sect.~3.1]{DLP17}.

\newcommand{\AlgU}{Algorithm}
\newcommand{\Alg}{Alg.}
\newcommand{\App}{App.}

\section{Equivalence of the Two Mechanisms}
\label{app:equivalence}

\newcommand{\MenPref}{a}
\newcommand{\WomenPref}{b}
\newcommand{\ManPref}[2]{\MenPref_{#1, #2}}
\newcommand{\ManPrefS}{\ManPref{\Man}{\Woman}}
\newcommand{\ManPrefUnmatched}[1]{\ManPref{#1}{\Unmatched}}
\newcommand{\ManPrefUnmatchedS}{\ManPrefUnmatched{\Man}}
\newcommand{\WomanPref}[2]{\WomenPref_{#1, #2}}
\newcommand{\WomanPrefS}{\WomanPref{\Man}{\Woman}}
\newcommand{\WomanPrefUnmatched}[1]{\WomanPref{\Unmatched}{#1}}
\newcommand{\WomanPrefUnmatchedS}{\WomanPrefUnmatched{\Woman}}

\newcommand{\MenPri}{\pi}
\newcommand{\ManPri}[1]{\MenPri_{#1}}
\newcommand{\ManPriS}{\ManPri{\Man}}

\newcommand{\MenUtil}{u}
\newcommand{\MenUtilp}{\MenUtil'}
\newcommand{\MenUtilpp}{\MenUtil''}
\newcommand{\MenUtilOpt}{\overline{\MenUtil}}
\newcommand{\MenUtilOptp}{\MenUtilOpt'}
\newcommand{\WomenUtil}{v}
\newcommand{\WomenUtilp}{\WomenUtil'}
\newcommand{\WomenUtilpp}{\WomenUtil''}
\newcommand{\WomenUtilPes}{\underline{\WomenUtil}}
\newcommand{\WomenUtilPesp}{\WomenUtilPes'}

\newcommand{\ManUtil}[1]{\UtilOrResOf{\MenUtil}{#1}}
\newcommand{\ManUtilp}[1]{\UtilOrResOf{\MenUtilp}{#1}}
\newcommand{\ManUtilOpt}[1]{\UtilOrResOf{\MenUtilOpt}{#1}}
\newcommand{\ManUtilOptp}[1]{\UtilOrResOf{\MenUtilOptp}{#1}}
\newcommand{\ManUtilS}{\ManUtil{\Man}}
\newcommand{\WomanUtil}[1]{\UtilOrResOf{\WomenUtil}{#1}}
\newcommand{\WomanUtilS}{\WomanUtil{\Woman}}

\newcommand{\WomenRes}{s}
\newcommand{\ManRespp}[1]{\UtilOrResOf{\MenRespp}{#1}}
\newcommand{\ManResS}{\ManRes{\Man}}
\newcommand{\WomanRes}[1]{\UtilOrResOf{\WomenRes}{#1}}
\newcommand{\WomanResS}{\WomanRes{\Woman}}

\newcommand{\ManComp}  [3]{f_{#1, #2}(\UtilOrResOf{#3}{#1})}
\newcommand{\WomanComp}[3]{g_{#2, #1}(\UtilOrResOf{#3}{#1})}
\newcommand{\ManCompS}  {\ManComp{\Man}{\Woman}{\MenUtil}}
\newcommand{\WomanCompS}{\WomanComp{\Woman}{\Man}{\WomenUtil}}
\newcommand{\ManCompMatchS}  [1]{\ManComp{\Man}{\MOf{#1}{\Man}}{\MenUtil}}
\newcommand{\WomanCompMatchS}[1]{\WomanComp{\Woman}{\MOf{#1}{\Woman}}{\WomenUtil}}

\newcommand{\SMMTuple}{(\Men, \Women, (\ManOrdPrefS)_{\EachMan}, (\WomanOrdPrefS)_{\EachWoman})}

\newcommand{\SMMC}{Stable marriage market} 
\PreviewMacro{\SMMC}
\newcommand{\SMM}{stable marriage market}
\PreviewMacro{\SMM}
\newcommand{\SMMa}{a \SMM}
\PreviewMacro{\SMMa}
\newcommand{\SMMs}{\SMM s}
\PreviewMacro{\SMMs}

\newcommand{\NN}{N}
\newcommand{\BB}{\lambda}

\newcommand{\TSMTuple}{(\Men, \Women, \MenPri, \NN, \BB, \MenPref, \WomenPref)}

\newcommand{\TSMC}{Tiered-slope market} 
\PreviewMacro{\TSMC}
\newcommand{\TSMCAll}{Tiered-Slope Market} 
\PreviewMacro{\TSMCAll}
\newcommand{\TSM}{tiered-slope market}
\PreviewMacro{\TSM}
\newcommand{\TSMa}{a \TSM}
\PreviewMacro{\TSMa}
\newcommand{\TSMs}{\TSM s}
\PreviewMacro{\TSMs}

\newcommand{\Market}{\mathcal{M}}
\newcommand{\Marketp}{\Market'}

In this appendix, we fix \SMMa\ $\SMMTuple$ and an associated \TSM\ $\Market = \TSMTuple$.
As in Sections~\ref{sec:market} and~\ref{sec:marriage}, $\ManCompS$ denotes the compensation that man $\Man$ needs to receive in order to attain utility $\ManUtilS$ in $\Market$ when he is matched to woman $\Woman$, and $\WomanCompS$ denotes the compensation that woman $\Woman$ needs to receive in order to attain utility $\WomanUtilS$ in $\Market$ when she is matched to man $\Man$.
We let $\MenRes$ denote the reserve utility vector of the men in $\Market$, and we let $\WomenRes$ denote the reserve utility vector of the women in $\Market$.
We consider an execution of \Alg~\ref{alg:smiw} on the \SMM\ $\SMMTuple$, and we let $\IAuc = (\MBidders, \Items)$ denote the IUAP constructed at line~\ref{line:smiw-iuap} of this execution.
We assume that $\IAuc$ is constructed in such a way that the following conditions hold:
for each man $\Man$, the priority of $\toMBidder{\Man}$ is $\ManPriS$;
the offers of the multibidders in $\IAuc$, i.e., the weights of the edges of $\IAuc$, satisfy the conditions stated in the last paragraph of App.~\ref{app:iuap-weights} below.
With the assumption that these conditions hold, we show in Theorem~\ref{thm:equivalence} that the set of greedy MWMs of $\UAP{\IAuc}$ corresponds to the set of man-optimal matchings of $\Market$.

\newcommand{\AuthorsDemangeGale}{Demange and Gale}
\newcommand{\AuthorsRothSotomayor}{Roth and Sotomayor}

\AlgU~\ref{alg:smiw} computes a greedy MWM of the UAP $\UAP{\IAuc}$ in lines~\ref{line:smiw-uap} and~\ref{line:smiw-greedy}.
Recall that we defined $\UAP{\IAuc}$ in \App~\ref{app:iuap} by giving an algorithm that converts an IUAP to a UAP, namely \Alg~\ref{alg:to-uap}.
In this appendix, we analyze the executions of \Alg~\ref{alg:to-uap} with input $\IAuc$ in order to relate the greedy MWMs of $\UAP{\IAuc}$ that \Alg~\ref{alg:smiw} computes to the man-optimal matchings of $\Market$.
Our approach is based on a technique used by~\AuthorsDemangeGale~\cite{DG85} to study various structural properties of their model, such as the lattice property.
\AuthorsDemangeGale\ analyze market instances in which the agents and their utility functions are fixed, while the reserve utilities vary.
As noted by \AuthorsRothSotomayor~\cite[Chapter 9]{roth+s:match}, lowering the reserve utility of an agent is analogous to extending the preferences of an agent in the stable marriage model, a technique used to study structural properties of the stable marriage model.
Building on this idea, for each iteration of Alg.~\ref{alg:to-uap}, we inductively show a bijection (Lemmas~\ref{lem:man-optimal-greedy-mwm} and \ref{lem:popt}) from the set of greedy MWMs of the UAP maintained at that iteration to the man-optimal matchings of the corresponding \TSM, where the reserve utilities are adjusted to ``reveal'' only the preferences that are present in the UAP.

\newcommand{\MarPayoff}[2]{(#1, #2)}
\newcommand{\MarPayoffS}{\MarPayoff{\MenUtil}{\WomenUtil}}
\newcommand{\MarPayoffSp}{\MarPayoff{\MenUtilp}{\WomenUtilp}}
\newcommand{\MarPayoffSpp}{\MarPayoff{\MenUtilpp}{\WomenUtilpp}}

\newcommand{\MarOutcome}[3]{(#1, #2, #3)}
\newcommand{\MarOutcomeS}{\MarOutcome{\MarMatch}{\MenUtil}{\WomenUtil}}

In the preceding sections, the terms ``feasible'', ``individually rational'', ``stable'', and ``man-optimal'' are used only for outcomes.
Throughout this appendix, however, we also use these terms for payoffs and matchings, as in~\cite{DG85}.
Here we briefly review the related definitions.
A pair $\MarPayoffS$ consisting of a utility vector $\MenUtil$ of the men and a utility vector $\WomenUtil$ of the women is a \emph{payoff}.
For any feasible outcome $\MarOutcomeS$ of a market $\Marketp$, we say that \emph{$\MarPayoffS$ is a feasible payoff of $\Marketp$}, and that \emph{$\MarMatch$ is compatible with $\MarPayoffS$}.
A feasible payoff $\MarPayoffS$ is \emph{individually rational} if $\ManUtilS \geq \ManResS$ for each man $\Man$ and $\WomanUtilS \geq \WomanResS$ for each woman $\Woman$.
If an outcome $\MarOutcomeS$ is stable (resp., man-optimal) for a market $\Marketp$, then we say that \emph{$\MarMatch$ is a stable} (resp., \emph{man-optimal}) \emph{matching of $\Marketp$}, and that \emph{$\MarPayoffS$ is a stable} (resp., \emph{man-optimal}) \emph{payoff in $\Marketp$}.
For a given market, there is a unique man-optimal payoff, but there can be more than one man-optimal matching.

\subsection{Edge Weights of the IUAP}
\label{app:iuap-weights}

\newcommand{\toWoman}[1]{\Id{woman}(#1)}
\newcommand{\toWomen}[1]{\Id{women}(#1)}
\newcommand{\toBidder}[2]{\Id{bidder}(#1, #2)}

We start our discussion by introducing some useful mappings from items to women and from man-woman pairs to bidders.
Recall that, for each woman $\Woman$, $\toItem{\Woman}$ is constructed at line~\ref{line:smiw-item} of Alg.~\ref{alg:smiw}, and for each man $\Man$, $\toDummy{\Man}$ is constructed at line~\ref{line:smiw-dummy} and $\toMBidder{\Man}$ is constructed at line~\ref{line:smiw-bidder}.
For any non-dummy item $\Item$ in $\Items$, we define $\toWoman{\Item}$ as the woman in $\Women$ associated with $\Item$.
(Thus for each woman $\Woman$, $\toWoman{\toItem{\Woman}}$ is equal to $\Woman$.)
For any dummy item $\Item$ in $\Items$, we define $\toWoman{\Item}$ as $\Unmatched$.
(Thus for each man $\Man$, $\toWoman{\toDummy{\Man}}$ is equal to $\Unmatched$.)
For any subset $\Itemsp$ of $\Items$, we define $\toWomen{\Itemsp}$ as $\SetBuild{\toWoman{\Item}}{\Item \in \Itemsp}$.
For any man $\Man$ in $\Men$ and any element $\Woman$ in $\Women + \Unmatched$, we define $\toBidder{\Man}{\Woman}$ as the bidder in $\toMBidder{\Man}$ such that $\toWomen{\BItems{\toBidder{\Man}{\Woman}}}$ contains $\Woman$.
Remark: For each man $\Man$, the dummy item $\toDummy{\Man}$ belongs to $\BItems{\toBidder{\Man}{\Unmatched}}$.

\newcommand{\IsMatched}[2]{\MOf{#2}{#1} \neq \Unmatched}
\newcommand{\IsUnmatched}[2]{\MOf{#2}{#1} = \Unmatched}
\newcommand{\IsMatchedS}[1]{\IsMatched{#1}{\MarMatch}}
\newcommand{\IsUnmatchedS}[1]{\IsUnmatched{#1}{\MarMatch}}

\newcommand{\WMW}[2]{\WId(#1, #2)}

In the remainder of the paper, for any man $\Man$ and any woman $\Woman$, we use the shorthand $\WMW{\Man}{\Woman}$ to denote $\WPair{\toBidder{\Man}{\Woman}}{\toItem{\Woman}}$, and $\WMW{\Man}{\Unmatched}$ to denote $\WPair{\toBidder{\Man}{\Unmatched}}{\toDummy{\Man}}$.

\newcommand{\WMPref}[1]{\WomenPref(#1)}
\newcommand{\WMUAP}[1]{\WId(#1)}

For any matching $\MarMatch$ of $\Market$, we define $\WMPref{\MarMatch}$ as
\begin{math}
  \sum_{\IsMatchedS{\Woman}} \WomanPref{\MOf{\MarMatch}{\Woman}}{\Woman} + \sum_{\IsUnmatchedS{\Woman}} \WomanPrefUnmatchedS
\end{math}, and we define $\WMUAP{\MarMatch}$ as
\begin{math}
  \sum_{\IsMatchedS{\Woman}} \WMW{\MOf{\MarMatch}{\Woman}}{\Woman}
\end{math}.

\newcommand{\CondUnmatchedMan}{(i)}
\newcommand{\CondMatching}{(ii)}

We assume that the IUAP $\IAuc$ is constructed so that the $\WMW{\Man}{\Woman}$'s satisfy the following conditions:
\CondUnmatchedMan\ for any man $\Man$, $\WMW{\Man}{\Unmatched} = 0$;
\CondMatching\ for any two matchings $\MarMatch$ and $\MarMatchp$ of $\Market$, $\WMUAP{\MarMatch} \geq \WMUAP{\MarMatchp}$ if and only if $\WMPref{\MarMatch} \geq \WMPref{\MarMatchp}$.
It is easy to see that one way to satisfy these conditions is to set $\WomanUtilFun{\Woman}{\Man} = \WomanPrefS - \WomanPrefUnmatchedS$ in \Alg~\ref{alg:smiw}.

\subsection{Matchings of the \TSMCAll\ and Greedy MWMs}
\label{app:tsm-matchings-greedy-mwms}

In this appendix, we first introduce certain mappings from the \TSM\ matchings to the UAP matchings, and we define weights for these \TSM\ matchings.
Then, in Lemma~\ref{lem:weight-greedy-mwm}, we show that these mappings define bijections from certain sets of \TSM\ matchings --- those maximizing the weights we introduce --- to the sets of greedy MWMs.
We start with some useful definitions.

\newcommand{\PRevealed}[2]{P_\Id{all}(#1, #2)}

For any configuration $\Conf = (\Auc, \IAuc)$ and any man $\Man$, we define the predicate $\PRevealed{\Conf}{\Man}$ to hold if $\Man$ has revealed all of his acceptable tiers in $\Conf$, i.e., if $\toBidder{\Man}{\Unmatched}$ belongs to $\Auc$.

\begin{lemma}
  \label{lem:all-revealed-matched}
  Let $\Conf = (\Auc, \IAuc)$ be a configuration and let $\Man$ be a man such that $\PRevealed{\Conf}{\Man}$ holds.
  Then for each greedy MWM $\Match$ of $\Auc$, some bidder associated with $\Man$ is matched in $\Match$.
\end{lemma}

\begin{proof}
  Suppose the claim does not hold, and let $\Match$ be a greedy MWM of $\Auc$ such that there is no bidder associated with $\Man$ that is matched in $\Match$.
  Then, since $\toBidder{\Man}{\Unmatched}$ and $\toDummy{\Man}$ belong to $\Auc$, $\WMW{\Man}{\Unmatched} = 0$, and $\BPri{\toBidder{\Man}{\Unmatched}} > 0$, we deduce that the matching $\Matchp = \Match + (\toBidder{\Man}{\Unmatched}, \toDummy{\Man})$ is an MWM of $\Auc$ such that $\MPriority{\Matchp} > \MPriority{\Match}$, a contradiction.
  \QEDWrap
\end{proof}

\newcommand{\ConfFirst}{\ConfSub{1}}
\newcommand{\ConfFinal}{\ConfSub{F}}
\newcommand{\CanE}{canonical}
\newcommand{\RelC}{relevant}

We define $\ConfFirst$ as the configuration $(\Auc, \IAuc)$ where $\Auc$ is the UAP that reveals only the first bidder of each multibidder in $\MBidders$.
We say that an execution of Alg.~\ref{alg:to-uap} invoked with input $\IAuc$ is \emph{\CanE} if $\ConfFirst$ is equal to the configuration that the program variable $\Conf$ stores at some iteration of this execution.
We define $\ConfFinal$ as the unique final configuration of any \CanE\ execution, i.e., $(\UAP{\IAuc}, \IAuc)$.
We say that a configuration is \emph{\RelC} if it is equal to the configuration that the program variable $\Conf$ stores at iteration $t_1$ or a subsequent iteration of some \CanE\ execution, where $t_1$ is the iteration at which $\ConfFirst = \Conf$.

\newcommand{\toUAP}[2]{\Phi_{#1}(#2)}
\newcommand{\toUAPFunc}[1]{\Phi_{#1}}

\newcommand{\ManPriUnmatched}[3]{\MenPri_{\Unmatched}(#1, #2, #3)}
\newcommand{\MenPriUnmatched}[2]{\MenPri_{\Unmatched}(#1, #2)}
\newcommand{\ManPriUnmatchedS}[1]{\ManPriUnmatched{\Conf}{#1}{\Man}}
\newcommand{\MenPriUnmatchedS}[1]{\MenPriUnmatched{\Conf}{#1}}

\newcommand{\WMarM}[2]{W_{#1}(#2)}

\newcommand{\ResC}[1]{(1 - \BB^{-1})}
\newcommand{\ResCS}{\ResC{\Man}}

Given a \RelC\ configuration $\Conf = (\Auc, \IAuc)$, we now introduce a mapping from certain matchings in the \TSM\ $\Market$ to the matchings in the UAP $\Auc$, and a weight function for these matchings in $\Market$.
Let $\Conf = (\Auc, \IAuc)$ be a \RelC\ configuration and let $\MarMatch$ be a matching of $\Market$ such that for each man-woman pair $(\Man, \Woman)$ matched in $\MarMatch$, $\toBidder{\Man}{\Woman}$ belongs to $\Auc$.
Then, we define $\toUAP{\Conf}{\MarMatch}$ as the matching
\begin{equation*}
  \bigcup_{\IsMatchedS{\Man}} \Set{(\toBidder{\Man}{\MOf{\MarMatch}{\Man}}, \toItem{\MOf{\MarMatch}{\Man}})} \cup \bigcup_{\IsUnmatchedS{\Man} \land \PRevealed{\Conf}{\Man}} \Set{(\toBidder{\Man}{\Unmatched}, \toDummy{\Man})} .
\end{equation*}
It is easy to see that $\toUAPFunc{\Conf}$ is an injection and that $\toUAP{\Conf}{\MarMatch}$ is a matching of the UAP $\Auc$.
Furthermore, we have $\WMatch{\toUAP{\Conf}{\MarMatch}} = \WMUAP{\MarMatch}$ since $\WPair{\toBidder{\Man}{\Unmatched}}{\toDummy{\Man}} = 0$ for any man $\Man$ by condition $\CondUnmatchedMan$ of App.~\ref{app:iuap-weights}.
For any man $\Man$, we define
\begin{equation*}
  \ManPriUnmatchedS{\MarMatch} =
  \begin{cases}
    0 & \text{if $\IsMatchedS{\Man}$} \\
    \ManPriS & \text{if $\IsUnmatchedS{\Man}$ and $\PRevealed{\Conf}{\Man}$} \\
    \ResCS & \text{if $\IsUnmatchedS{\Man}$ and $\lnot \PRevealed{\Conf}{\Man}$} .
  \end{cases}
\end{equation*}
We define $\MenPriUnmatchedS{\MarMatch}$ as $\sum_{\EachMan} \ManPriUnmatchedS{\MarMatch}$.
Remark: It is easy to see that
\begin{equation*}
  \MenPriUnmatchedS{\MarMatch} = \sum_{\IsUnmatchedS{\Man}} \ManPriUnmatchedS{\MarMatch} = \sum_{\IsUnmatchedS{\Man} \land \PRevealed{\Conf}{\Man}} \ManPriS + \sum_{\IsUnmatchedS{\Man} \land \lnot \PRevealed{\Conf}{\Man}} \ResCS .
\end{equation*}
Finally, we define $\WMarM{\Conf}{\MarMatch}$ as
\begin{equation*}
  \NN \cdot \WMPref{\MarMatch} + \sum_{\IsMatchedS{\Man}} \ManPriS + \MenPriUnmatchedS{\MarMatch} .
\end{equation*}
Remark: It is easy to see that
\begin{equation}
  \label{eq:WMarM}
  \WMarM{\Conf}{\MarMatch} = \NN \cdot \WMPref{\MarMatch} + \sum_{\IsMatchedS{\Man} \lor \PRevealed{\Conf}{\Man}} \ManPriS + \sum_{\IsUnmatchedS{\Man} \land \lnot \PRevealed{\Conf}{\Man}} \ResCS ,
\end{equation}
and that the second term in the RHS is equal to $\MPriority{\toUAP{\Conf}{\MarMatch}}$.

We now show that $\toUAPFunc{\Conf}$ is invertible in certain cases.

\newcommand{\CondInverseAtMostOne}{(1)}
\newcommand{\CondInverseRevealedMatched}{(2)}

\begin{lemma}
  \label{lem:inverse-of-good-matchings}
  Let $\Conf = (\Auc, \IAuc)$ be a \RelC\ configuration and let $\Match$ be a matching of $\Auc$ such that the following conditions hold:
  \CondInverseAtMostOne\ for each man $\Man$, at most one bidder associated with $\Man$ is matched in $\Match$;
  \CondInverseRevealedMatched\ for each man $\Man$, if $\PRevealed{\Conf}{\Man}$ holds then a bidder associated with $\Man$ is matched in $\Match$.
  Then there exists a unique matching $\MarMatch$ of $\Market$ such that $\toUAP{\Conf}{\MarMatch}$ is equal to $\Match$.
\end{lemma}

\begin{proof}
  Let $\MarMatch$ denote the matching such that $\MOf{\MarMatch}{\Man} = \Woman$ if and only if $(\toBidder{\Man}{\Woman}, \toItem{\Woman})$ belongs to $\Match$.
  It is easy to see by condition~\CondInverseAtMostOne\ that $\MarMatch$ is a matching of $\Market$.
  Moreover, condition~\CondInverseRevealedMatched\ implies that if $\MOf{\MarMatch}{\Man} = \Unmatched$ for a man $\Man$, then $(\toBidder{\Man}{\Unmatched}, \toDummy{\Man})$ belongs to $\Match$.
  The claim follows since $\toUAPFunc{\Conf}$ is an injection and $\toUAP{\Conf}{\MarMatch}$ is equal to $\Match$.
  \QEDWrap
\end{proof}

\begin{lemma}
  \label{lem:inverse-of-greedy-mwms}
  Let $\Conf = (\Auc, \IAuc)$ be a \RelC\ configuration and let $\Match$ be a greedy MWM of $\Auc$.
  Then there exists a unique matching $\MarMatch$ of $\Market$ such that $\toUAP{\Conf}{\MarMatch}$ is equal to $\Match$.
\end{lemma}

\begin{proof}
  It is sufficient to prove that $\Match$ satisfies the two conditions of Lemma~\ref{lem:inverse-of-good-matchings}.
  Condition~\CondInverseAtMostOne\ is satisfied because Lemmas~\ref{lem:iuap-one} and \ref{lem:iuap-tail} imply that, for each man $\Man$, $\Match$ matches at most one bidder associated with $\Man$.
  Lemma~\ref{lem:all-revealed-matched} implies that Condition~\CondInverseRevealedMatched\ is satisfied.
  \QEDWrap
\end{proof}

We now show that, given a \RelC\ configuration $\Conf = (\Auc, \IAuc)$, $\toUAPFunc{\Conf}$ is a bijection from a certain set of matchings of $\Market$ to the set of greedy MWMs of $\Auc$.
We start with some useful lemmas and definitions that help us to define this set.

\begin{lemma}
  \label{lem:reveals-acceptable}
  Let $\Conf = (\Auc, \IAuc)$ be a \RelC\ configuration and let $(\Man, \Woman)$ be a man-woman pair such that $\toBidder{\Man}{\Woman}$ belongs to $\Auc$.
  Then $\ManPrefS \geq \ManPrefUnmatchedS$.
\end{lemma}

\begin{proof}
  If the least preferred acceptable tier of a man $\Man$ is revealed at some iteration, i.e., if $\toBidder{\Man}{\Unmatched}$ is added at line~\ref{line:reveal} of \Alg~\ref{alg:to-uap}, then for any configuration $\Confp$ that results in a subsequent iteration, $\PRevealed{\Confp}{\Man}$ holds, and by Lemma~\ref{lem:all-revealed-matched}, there is no bidder associated with $\Man$ in $\CReady{\Confp}$.
  Hence no other tier of $\Man$ is subsequently revealed.
  \QEDWrap
\end{proof}

\newcommand{\LeastTier}[2]{\Id{least}(#1, #2)}

For any \RelC\ configuration $\Conf$ and any man $\Man$, we define $\LeastTier{\Conf}{\Man}$ as the nonempty subset of $\Women + \Unmatched$ in the least preferred tier of $\Man$ that is revealed in $\Conf$, i.e.,  $\toWomen{\BItems{\MBBidder{\MBidder}{\Car{\BiddersOfMB{\Conf}{\MBidder}}}}}$, where $\MBidder$ denotes $\toMBidder{\Man}$.
Remark: $\ManPrefS = \ManPref{\Man}{\Womanp} \geq \ManPrefUnmatchedS$ for any $\Woman$ and $\Womanp$ belonging to $\LeastTier{\Conf}{\Man}$, where the inequality follows from Lemma~\ref{lem:reveals-acceptable} and is tight if and only if $\PRevealed{\Conf}{\Man}$ holds.

\newcommand{\PLeast}[2]{P_{\Id{least}}(#1, #2)}

For any \RelC\ configuration $\Conf$ and any matching $\MarMatch$ of $\Market$, we define the predicate $\PLeast{\Conf}{\MarMatch}$ to hold if for each man-woman pair $(\Man, \Woman)$ matched in $\MarMatch$, the woman $\Woman$ belongs to $\LeastTier{\Conf}{\Man}$.
Remark: For any \RelC\ configuration $\Conf = (\Auc, \IAuc)$ and any matching $\MarMatch$ such that $\PLeast{\Conf}{\MarMatch}$ holds, it is easy to see that $\toUAP{\Conf}{\MarMatch}$ is well-defined because for each man-woman pair $(\Man, \Woman)$ matched in $\MarMatch$, $\toBidder{\Man}{\Woman}$ belongs to $\Auc$.

The following lemma is only used to prove Lemma~\ref{lem:weight-greedy-mwm}.

\begin{lemma}
  \label{lem:uap-weight-marm}
  Let $\Conf = (\Auc, \IAuc)$ be a \RelC\ configuration, and let $\MarMatchSub{0}$ and $\MarMatch$ be two matchings of $\Market$ such that $\PLeast{\Conf}{\MarMatchSub{0}}$ holds, $\toUAP{\Conf}{\MarMatch}$ is a greedy MWM of $\Auc$, and $\WMUAP{\MarMatchSub{0}} < \WMUAP{\MarMatch}$.
  Then $\WMarM{\Conf}{\MarMatchSub{0}} < \WMarM{\Conf}{\MarMatch}$.
\end{lemma}

\newcommand{\ic}{k} 

\begin{proof}
  \newcommand{\EdgesInitial}[1]{X'_{#1}}
  \newcommand{\EdgesFinal}[1]{X_{#1}}

  Since $\toUAP{\Conf}{\MarMatch}$ is a greedy MWM of $\Auc$, Lemma~\ref{lem:iuap-tail} implies that $\PLeast{\Conf}{\MarMatch}$ holds.
  Let $\MatchSub{0}$ denote $\toUAP{\Conf}{\MarMatchSub{0}}$.
  The symmetric difference of $\MatchSub{0}$ and $\Match$, denoted $\MatchSub{0} \oplus \Match$, corresponds to a collection $\PathCycleColl$ of vertex-disjoint paths and cycles.
  Let $\Seq{\PathCycleSub{1}, \dotsc, \PathCycleSub{\Car{\PathCycleColl}}}$ be an arbitrary permutation of $\PathCycleColl$.
  For any integer $\ic$ such that $1 \leq \ic \leq \Car{\PathCycleColl}$,
  let $\EdgesFinal{\ic}$ denote the edges of $\PathCycleSub{\ic}$ that belong to $\Match$,
  let $\EdgesInitial{\ic}$ denote the edges of $\PathCycleSub{\ic}$ that belong to $\MatchSub{0}$, and
  let $\MatchSub{\ic}$ denote the matching $(\MatchSub{\ic-1} \setminus \EdgesInitial{\ic}) \cup \EdgesFinal{\ic}$.
  Remark: It is easy to see that $\MatchSub{\Car{\PathCycleColl}}$ is equal to $\Match$.

  We start by showing that, for each integer $\ic$ such that $1 \leq \ic \leq \Car{\PathCycleColl}$, $\MatchSub{\ic}$ satisfies the two conditions of Lemma~\ref{lem:inverse-of-good-matchings}.
  It is easy to see that condition~\CondInverseAtMostOne\ holds because $\PLeast{\Conf}{\MarMatchSub{0}}$ and $\PLeast{\Conf}{\MarMatch}$ imply that, for any integer $\ic$ such that $0 \leq \ic \leq \Car{\PathCycleColl}$, any bidder that is matched in $\MatchSub{\ic}$ is the least preferred bidder of the associated man that is revealed in $\Auc$.
  We now address condition~\CondInverseRevealedMatched.
  Since $\PLeast{\Conf}{\MarMatchSub{0}}$ and $\PLeast{\Conf}{\MarMatch}$ hold, the definitions of $\MatchSub{0}$ and $\Match$ imply that for each man $\Man$ such that $\PRevealed{\Conf}{\Man}$ holds, both $\MatchSub{0}$ and $\Match$ match the bidder $\toBidder{\Man}{\Unmatched}$.
  It follows that, for each man $\Man$ such that $\PRevealed{\Conf}{\Man}$ holds, $\toBidder{\Man}{\Unmatched}$ is not an endpoint of any path in $\PathCycleColl$, and thus $\toBidder{\Man}{\Unmatched}$ is matched in $\MatchSub{\ic}$ for all $1 \leq \ic \leq \Car{\PathCycleColl}$, establishing condition~\CondInverseRevealedMatched.
  Having established that $\MatchSub{\ic}$ satisfies the two conditions of Lemma~\ref{lem:inverse-of-good-matchings}, for each integer $\ic$ such that $1 \leq \ic \leq \Car{\PathCycleColl}$, we define $\MarMatchSub{\ic}$ as the matching of $\Market$ such that $\toUAP{\Conf}{\MarMatchSub{\ic}}$ is equal to $\MatchSub{\ic}$.
  We now establish two simple but useful claims.

  Claim 1: $\WMatch{\MatchSub{\ic}} > \WMatch{\MatchSub{\ic-1}}$ for at least one $\ic$ such that $1 \leq \ic \leq \Car{\PathCycleColl}$.
  The claim follows directly from the fact that $\WMatch{\MatchSub{0}} = \WMUAP{\MarMatchSub{0}} < \WMUAP{\MarMatch} = \WMatch{\Match} = \WMatch{\MatchSub{\Car{\PathCycleColl}}}$.

  Claim 2: $\WMatch{\MatchSub{\ic}} \geq \WMatch{\MatchSub{\ic-1}}$ for all $1 \leq \ic \leq \Car{\PathCycleColl}$.
  For the sake of contradiction, suppose that the claim fails for some integer $\ic$.
  Then $(\Match \setminus \EdgesFinal{\ic}) \cup \EdgesInitial{\ic}$ is a matching of $\Auc$ with weight higher than that of $\Match$, a contradiction since $\Match$ is an MWM of $\Auc$.

  Having established these two claims, we now complete the proof of the lemma by showing that the following two conditions hold for any $1 \leq \ic \leq \Car{\PathCycleColl}$:
  (a) if $\WMatch{\MatchSub{\ic}} > \WMatch{\MatchSub{\ic-1}}$ then $\WMarM{\Conf}{\MarMatchSub{\ic}} > \WMarM{\Conf}{\MarMatchSub{\ic-1}}$; and
  (b) if $\WMatch{\MatchSub{\ic}} = \WMatch{\MatchSub{\ic-1}}$ then $\WMarM{\Conf}{\MarMatchSub{\ic}} \geq \WMarM{\Conf}{\MarMatchSub{\ic-1}}$.

  We first address condition~(a).
  Let $\ic$ be an integer such that $1 \leq \ic \leq \Car{\PathCycleColl}$ and $\WMatch{\MatchSub{\ic}} > \WMatch{\MatchSub{\ic-1}}$.
  Our goal is to establish that $\WMarM{\Conf}{\MarMatchSub{\ic}} > \WMarM{\Conf}{\MarMatchSub{\ic-1}}$.
  Since $\Car{\Men} - \Car{\toUAP{\Conf}{\MarMatchp}}$ is equal to $\Car{\SetBuild{\Man}{\IsUnmatched{\Man}{\MarMatchp} \land \lnot \PRevealed{\Conf}{\Man}}}$ for any $\MarMatchp$, equality~\eqref{eq:WMarM} and the associated remark imply that the difference $\WMarM{\Conf}{\MarMatchSub{\ic}} - \WMarM{\Conf}{\MarMatchSub{\ic-1}}$ is equal to
  \begin{equation}
    \label{eq:successive-weight-diff}
    \begin{multlined}
      \NN \cdot (\WMPref{\MarMatchSub{\ic}} - \WMPref{\MarMatchSub{\ic-1}}) + (\MPriority{\MatchSub{\ic}} - \MPriority{\MatchSub{\ic-1}}) + {}\\
      (\Car{\MatchSub{\ic-1}} - \Car{\MatchSub{\ic}}) \ResCS .
    \end{multlined}
  \end{equation}
  Since $\WMUAP{\MarMatchSub{\ic}} = \WMatch{\MatchSub{\ic}} > \WMatch{\MatchSub{\ic-1}} = \WMUAP{\MarMatchSub{\ic-1}}$, condition~\CondMatching\ stated in App.~\ref{app:iuap-weights} implies that $\WMPref{\MarMatchSub{\ic}} > \WMPref{\MarMatchSub{\ic-1}}$.
  Then, since $\NN$, $\WMPref{\MarMatchSub{\ic}}$, $\WMPref{\MarMatchSub{\ic-1}}$, and the priorities are integers, and since $\NN > \max_{\EachMan} \ManPriS$, we deduce that the first term of~\eqref{eq:successive-weight-diff} is at least $1 + \max_{\EachMan} \ManPriS$.
  Thus, in order to establish that $\WMarM{\Conf}{\MarMatchSub{\ic}} > \WMarM{\Conf}{\MarMatchSub{\ic-1}}$, it is enough to show that the sum of the second and third term of~\eqref{eq:successive-weight-diff} is greater than $-1 - \max_{\EachMan} \ManPriS$.
  If $\PathCycleSub{\ic}$ is a cycle, then it is easy to see that $\MatchedBidders{\MatchSub{\ic}} = \MatchedBidders{\MatchSub{\ic-1}}$, and hence that both the second and third terms of~\eqref{eq:successive-weight-diff} are zero.
  In the remainder of this paragraph, we address the case where $\PathCycleSub{\ic}$ is a path.
  In this case, it is easy to see that $\MatchedBidders{\MatchSub{\ic-1}} \setminus \MatchedBidders{\MatchSub{\ic}}$ contains at most one bidder.
  Then, since $\min_{\EachMan} \ManPriS > 0$, we deduce that the second term of~\eqref{eq:successive-weight-diff} is at least $-\max_{\EachMan} \ManPriS$.
  Finally, since $-1 \leq \Car{\MatchSub{\ic-1}} - \Car{\MatchSub{\ic}} \leq 1$, we conclude that the third term of~\eqref{eq:successive-weight-diff} is greater than $-1$, as required.
  
  We now address condition~(b).
  Let $\ic$ be an integer such that $1 \leq \ic \leq \Car{\PathCycleColl}$ and $\WMatch{\MatchSub{\ic}} = \WMatch{\MatchSub{\ic-1}}$.
  Our goal is to establish that $\WMarM{\Conf}{\MarMatchSub{\ic}} \geq \WMarM{\Conf}{\MarMatchSub{\ic-1}}$.
  Again, the difference $\WMarM{\Conf}{\MarMatchSub{\ic}} - \WMarM{\Conf}{\MarMatchSub{\ic-1}}$ is equal to~\eqref{eq:successive-weight-diff}.
  In this case, since $\WMUAP{\MarMatchSub{\ic}} = \WMatch{\MatchSub{\ic}} = \WMatch{\MatchSub{\ic-1}} = \WMUAP{\MarMatchSub{\ic-1}}$, condition~\CondMatching\ stated in App.~\ref{app:iuap-weights} implies that $\WMPref{\MarMatchSub{\ic}} = \WMPref{\MarMatchSub{\ic-1}}$, and hence that the first term of~\eqref{eq:successive-weight-diff} is zero.
  Thus, in order to establish that $\WMarM{\Conf}{\MarMatchSub{\ic}} \geq \WMarM{\Conf}{\MarMatchSub{\ic-1}}$, it remains to show that the sum of the second and third term of~\eqref{eq:successive-weight-diff} is nonnegative.
  If $\PathCycleSub{\ic}$ is a cycle, then it is easy to see that $\MatchedBidders{\MatchSub{\ic}} = \MatchedBidders{\MatchSub{\ic-1}}$, and hence that both the second and third terms of~\eqref{eq:successive-weight-diff} are zero.
  In the remainder of this paragraph, we address the case where $\PathCycleSub{\ic}$ is a path.
  In this case, it is easy to see that $\MatchedBidders{\MatchSub{\ic}} \not= \MatchedBidders{\MatchSub{\ic-1}}$.
  Since (as argued in the second paragraph of the proof) any bidder that is matched in $\MatchSub{\ic}$ or in $\MatchSub{\ic-1}$ is the least preferred bidder of the associated man that is revealed in $\Auc$, we deduce that $\MPriority{\MatchSub{\ic}} \not= \MPriority{\MatchSub{\ic-1}}$.
  We conclude that $\MPriority{\MatchSub{\ic}} > \MPriority{\MatchSub{\ic-1}}$, for otherwise $(\Match \setminus \EdgesFinal{\ic}) \cup \EdgesInitial{\ic}$ is an MWM of $\Auc$ with priority higher than that of $\Match$, a contradiction since $\Match$ is a greedy MWM of $\Auc$.
  It follows that the second term of~\eqref{eq:successive-weight-diff} is at least $1$ since the priorities are integers.
  Finally, since $-1 \leq \Car{\MatchSub{\ic-1}} - \Car{\MatchSub{\ic}} \leq 1$, we conclude that the third term of~\eqref{eq:successive-weight-diff} is greater than $-1$, as required.
  \QEDWrap
\end{proof}

\newcommand{\MarMatchSet}{X}
\newcommand{\MarMatchSetp}{\MarMatchSet'}
\newcommand{\MarMatchSetOpt}{\Opt{\MarMatchSet}}

\begin{lemma}
  \label{lem:weight-greedy-mwm}
  Let $\Conf = (\Auc, \IAuc)$ be a \RelC\ configuration and let $\MarMatchSet$ be the set of all matchings $\MarMatch$ of $\Market$ such that $\PLeast{\Conf}{\MarMatch}$ holds.
  Let $\MarMatchSetOpt$ denote the set $\SetBuild{\MarMatch}{\MarMatch \in \MarMatchSet \land \WMarM{\Conf}{\MarMatch} = \max_{\MarMatchp \in \MarMatchSet} \WMarM{\Conf}{\MarMatchp}}$.
  Then, $\toUAPFunc{\Conf}$ is a bijection from $\MarMatchSetOpt$ to the set of greedy MWMs of $\Auc$.
\end{lemma}

\newcommand{\Menp}{\Men'}

\newcommand{\Var}{x}

\begin{proof}
  Let $\Match$ be a greedy MWM of $\Auc$ and let $\MarMatch$ be the matching (by Lemma~\ref{lem:inverse-of-greedy-mwms}) of $\Market$ such that $\toUAP{\Conf}{\MarMatch}$ is equal to $\Match$.
  Lemma~\ref{lem:iuap-tail} implies that $\PLeast{\Conf}{\MarMatch}$ holds;
  thus $\MarMatch$ belongs to $\MarMatchSet$.
  Let $\MarMatchOpt$ be a matching in $\MarMatchSetOpt$ and let $\MatchOpt$ denote $\toUAP{\Conf}{\MarMatchOpt}$.
  Note that $\MatchOpt$ is a matching of $\Auc$.
  Since $\toUAPFunc{\Conf}$ is an injection, it remains to show that $\WMarM{\Conf}{\MarMatch} = \max_{\MarMatchp \in \MarMatchSet} \WMarM{\Conf}{\MarMatchp}$ and that $\MatchOpt$ is a greedy MWM of $\Auc$.

  Since $\WMarM{\Conf}{\MarMatch} \leq \max_{\MarMatchp \in \MarMatchSet} \WMarM{\Conf}{\MarMatchp} = \WMarM{\Conf}{\MarMatchOpt}$, Lemma~\ref{lem:uap-weight-marm} implies that $\WMUAP{\MarMatch} \leq \WMUAP{\MarMatchOpt}$.
  Then, since $\Match$ is an MWM of $\Auc$, we deduce that $\MatchOpt$ is an MWM of $\Auc$.
  Since $\Match$ is a greedy MWM, and hence an MCMWM of $\Auc$, we deduce that $\Car{\MatchOpt} \leq \Car{\Match}$.
  We consider two cases.

  Case 1: $\Car{\MatchOpt} < \Car{\Match}$.
  Let $\Var$ be a bidder in $\MatchedBidders{\Match} \setminus \MatchedBidders{\MatchOpt}$ such that there exists an MWM, call it $\Matchp$, of $\Auc$ having $\MatchedBidders{\Matchp} = \MatchedBidders{\MatchOpt} + \Var$;
  the exchange property of $\Matroid{\Auc}$ implies the existence of such a bidder $\Var$ and matching $\Matchp$.
  Let $\Manp$ denote the man associated with $\Var$.
  We first argue that $\Matchp$ satisfies the two conditions of Lemma~\ref{lem:inverse-of-good-matchings}.
  By definition, the men corresponding to the bidders in $\MatchedBidders{\MatchOpt}$ are distinct;
  let $\Menp$ denote the set of men corresponding to these bidders, i.e., $\Menp = \SetBuild{\Man}{\IsMatched{\Man}{\MarMatchOpt} \lor \PRevealed{\Conf}{\Man}}$.
  Since $\PLeast{\Conf}{\MarMatchOpt}$ holds, we deduce that, for each man $\Man$ in $\Menp$, the bidder in $\MatchedBidders{\MatchOpt}$ associated with $\Man$ corresponds to the least preferred tier of $\Man$ that is revealed in $\Conf$.
  Lemma~\ref{lem:iuap-tail} implies that $\Var$ corresponds to the least preferred tier of $\Manp$ that is revealed in $\Conf$.
  Then, since $\Var$ does not belong to $\MatchedBidders{\MatchOpt}$, the results of the preceding two sentences imply that $\Manp$ does not belong to $\Menp$, and thus that $\Matchp$ satisfies condition~\CondInverseAtMostOne.
  We now address condition~\CondInverseRevealedMatched.
  Matching $\MatchOpt$ satisfies condition~\CondInverseRevealedMatched\ by definition.
  Thus $\Matchp$ satisfies condition~\CondInverseRevealedMatched\ since any bidder matched by $\MatchOpt$ is also matched by $\Matchp$.
  Since $\Matchp$ satisfies conditions~\CondInverseAtMostOne\ and~\CondInverseRevealedMatched, Lemma~\ref{lem:inverse-of-good-matchings} implies that there is a matching, call it $\MarMatchp$, of $\Market$ such that $\toUAP{\Conf}{\MarMatchp}$ is equal to $\Matchp$.
  Since both $\Matchp$ and $\MatchOpt$ are MWMs of $\Auc$, condition~\CondMatching\ stated in \App~\ref{app:iuap-weights} implies that $\WMPref{\MarMatchp} = \WMPref{\MarMatchOpt}$.
  Since $\MatchedBidders{\Matchp}$ properly contains $\MatchedBidders{\MatchOpt}$, we deduce that the set $\SetBuild{\Man}{\IsMatched{\Man}{\MarMatchp} \lor \PRevealed{\Conf}{\Man}}$ properly contains the set $\SetBuild{\Man}{\IsMatched{\Man}{\MarMatchOpt} \lor \PRevealed{\Conf}{\Man}}$.
  Then, since $\ManPriS > \ResCS$ for each man $\Man$, the results of the preceding two sentences imply that $\WMarM{\Conf}{\MarMatchp} > \WMarM{\Conf}{\MarMatchOpt}$, contradicting the definition of $\MarMatchOpt$.

  Case 2: $\Car{\MatchOpt} = \Car{\Match}$.
  Since both $\Match$ and $\MatchOpt$ are MWMs of $\Auc$, condition~\CondMatching\ stated in \App~\ref{app:iuap-weights} implies that $\WMPref{\MarMatch} = \WMPref{\MarMatchOpt}$.
  Since $\Car{\MatchedBidders{\Match}} = \Car{\MatchedBidders{\MatchOpt}}$, we deduce that the cardinality of $\SetBuild{\Man}{\IsMatchedS{\Man} \lor \PRevealed{\Conf}{\Man}}$ is equal to the cardinality of $\SetBuild{\Man}{\IsMatched{\Man}{\MarMatchOpt} \lor \PRevealed{\Conf}{\Man}}$.
  Then, by using the results of the preceding two sentences and the remark regarding the second term in the RHS of~(\ref{eq:WMarM}), we deduce that $\WMarM{\Conf}{\MarMatch} - \WMarM{\Conf}{\MarMatchOpt} = \MPriority{\toUAP{\Conf}{\MarMatch}} - \MPriority{\toUAP{\Conf}{\MarMatchOpt}}$.
  The latter expression is nonnegative since $\Match$ is a greedy MWM of $\Auc$, and it is nonpositive since $\WMarM{\Conf}{\MarMatchOpt} = \max_{\MarMatchp \in \MarMatchSet} \WMarM{\Conf}{\MarMatchp}$.
  Thus we deduce that $\WMarM{\Conf}{\MarMatch} = \max_{\MarMatchp \in \MarMatchSet} \WMarM{\Conf}{\MarMatchp}$ and that $\MatchOpt$ is a greedy MWM of $\Auc$.
  \QEDWrap
\end{proof}

\subsection{Revealing Preferences in the \TSMCAll}
\label{app:tsm-reveal}

Recall that Alg.~\ref{alg:to-uap} iteratively reveals the bidders, which correspond to the tiers of men, and the state of the revealed bidders are captured in a configuration.
In this appendix, we first show how to adjust the reserve utilities of the men to obtain markets identical to the \TSM\ except that only the tiers that are revealed in a configuration are acceptable (Lemma~\ref{lem:reveal-reserve}).
Then, we inductively show a bijection (Lemma~\ref{lem:man-optimal-greedy-mwm} and \ref{lem:popt}) from the set of greedy MWMs of the UAP maintained at each iteration of Alg.~\ref{alg:to-uap} to the man-optimal matchings of the corresponding market with adjusted reserve utilities.
Finally, we establish our result in Theorem~\ref{thm:equivalence}, and we prove Theorem~\ref{thm:implementation} of Sect.~\ref{sec:implementation}.

\newcommand{\ConfRes}[1]{\Id{reserve}(#1)}

For any \RelC\ configuration $\Conf$, we define $\ConfRes{\Conf}$ as the reserve utility vector $\MenResp$ of $\Men$ such that for each man $\Man$,
\begin{equation*}
  \ManResp{\Man} = \max\cbr{\ManResS, \ResCS \lexp{\ManPref{\Man}{\Woman}}} =
  \begin{cases}
    \ManPriS \lexp{\ManPref{\Man}{\Unmatched}} & \text{if $\PRevealed{\Conf}{\Man}$} \\
    \ResCS \lexp{\ManPref{\Man}{\Woman}} & \text{otherwise} ,
  \end{cases}
\end{equation*}
where $\Woman$ is some element in $\LeastTier{\Conf}{\Man}$.

\newcommand{\MarketRes}[1]{\Market(#1)}
\newcommand{\MarketConf}[1]{\Market(#1)}

For any reserve utility vector $\MenResp$ of $\Men$ such that $\MenResp \geq \MenRes$, we define $\MarketRes{\MenResp}$ as the market that is equal to $\Market$ except that the reserve utilities of the men are given by $\MenResp$.
For any \RelC\ configuration $\Conf$, we define $\MarketConf{\Conf}$ as $\MarketRes{\ConfRes{\Conf}}$.
Lemma~\ref{lem:reveal-reserve} below shows that, for any \RelC\ configuration $\Conf$, only the tiers of men that are revealed in $\Conf$ are ``acceptable'' in $\MarketConf{\Conf}$.

\begin{lemma}
\label{lem:reveal-reserve}
  Let $\Conf = (\Auc, \IAuc)$ be a \RelC\ configuration, let $\MarOutcomeS$ be an individually rational outcome for $\MarketConf{\Conf}$, and let $(\Man, \Woman)$ be a man-woman pair matched in $\MarMatch$.
  Then $\toBidder{\Man}{\Woman}$ belongs to $\Auc$.
\end{lemma}

\begin{proof}
  We have
  \begin{align*}
    \ManCompS
    &\leq -\WomanCompS = -\WomanUtilS + \NN \WomanPrefS + \ManPriS \\
    &\leq \NN \pr{\WomanPrefS - \WomanPrefUnmatchedS } + \ManPriS \\
    &\leq \NN \WomanPrefS - \NN + \ManPriS \\
    &\leq \BB - 2\NN + \ManPriS \\
    &\leq \BB - 2 ,
  \end{align*}
  where the inequalities are justified as follows:
  the first inequality follows from the feasibility of $\MarOutcomeS$;
  the second inequality follows from the individual rationality of $\MarOutcomeS$, which implies $\WomanUtilS \geq \WomanResS = \NN \WomanPrefUnmatchedS$;
  the third inequality follows since $\WomanPrefUnmatchedS \geq 1$;
  the fourth inequality follows since $\BB \geq \max_{\Man, \Woman}(\WomanPrefS+1) \NN$;
  the fifth inequality follows since $\NN > \max_{\EachMan}\ManPriS$.
  Then, since $\ManCompS \leq \BB - 2$, we deduce that $\ManUtilS \leq (\BB - 2) \lexp{\ManPref{\Man}{\Woman}}$.

  Let $\MenResp$ denote $\ConfRes{\Conf}$ and let $\Womanp$ be an arbitrary element in $\LeastTier{\Conf}{\Man}$.
  Assume the claim of the lemma is false: thus $\ManPref{\Man}{\Woman} < \ManPref{\Man}{\Womanp}$.
  Then, since $\ManUtilS \leq (\BB - 2) \lexp{\ManPref{\Man}{\Woman}}$, we conclude that $\ManUtilS < \ResCS \lexp{\ManPref{\Man}{\Womanp}}$, and hence that $\ManUtilS < \ManResp{\Man}$, contradicting individual rationality.
  \QEDWrap
\end{proof}

The following lemma provides a lower bound on the utilities of men in the man-optimal outcomes.
It is used in the proof of Lemma~\ref{lem:popt}.

\newcommand{\MarPayoffMenOpt}{\MarPayoff{\MenUtilOpt}{\WomenUtilPes}}
\newcommand{\MarPayoffMenOptp}{\MarPayoff{\MenUtilOptp}{\WomenUtilPesp}}
\newcommand{\MarOutcomeMenOpt}[1]{\MarOutcome{#1}{\MenUtilOpt}{\WomenUtilPes}}

\begin{lemma}
  \label{lem:man-lower-bound}
  Let $\Conf$ be a \RelC\ configuration, let $\MarOutcomeMenOpt{\MarMatch}$ be a man-optimal outcome for $\MarketConf{\Conf}$, and let $(\Man, \Woman)$ be a man-woman pair matched in $\MarMatch$.
  Then $\ManUtilOpt{\Man} \geq \lexp{\ManPref{\Man}{\Woman}}$.
\end{lemma}

\begin{proof}
  Let $\MenResp$ denote $\ConfRes{\Conf}$.
  We show that $\MenResp$ satisfies the conditions required by Lemma~\ref{lem:slackness-general}, then the claim follows from that lemma.
  By definition, $\MenResp$ is at least $\MenRes$, so it satisfies condition~\CondMenResMinimum.
  For any man $\Man$, if $\PRevealed{\Conf}{\Man}$ holds, then $\ManResp{\Man}$ is equal to $\ManPriS \lexp{\ManPref{\Man}{\Unmatched}}$, and it is easy to see that $\ManResp{\Man} \lexp{k}$ is an integer for any integer $k \geq -\ManPrefUnmatchedS$ and that $0 < \ManResp{\Man} \lexp{k} < 1$ for any integer $k < -\ManPrefUnmatchedS$;
  otherwise, $\ManResp{\Man}$ is equal to $\ResCS \lexp{\ManPref{\Man}{\Woman}}$ where $\Woman$ is some element in $\LeastTier{\Conf}{\Man}$, and it is easy to see that $\ManResp{\Man} \lexp{k}$ is an integer for any integer $k > -\ManPref{\Man}{\Woman}$ and that $0 < \ManResp{\Man} \lexp{k} < 1$ for any integer $k \leq -\ManPref{\Man}{\Woman}$.
  Thus, $\MenResp$ satisfies condition~\CondMenResIntegrality\ as well.
  \QEDWrap
\end{proof}

For a matching $\MarMatch$ satisfying $\PLeast{\Conf}{\MarMatch}$, the following lemma gives a necessary and sufficient condition for $\MarMatch$ to be a stable matching of $\MarketConf{\Conf}$.
The proof of the lemma is quite involved and is deferred to App.~\ref{sec:equivalence-proof}.

\begin{lemma}
  \label{lem:max-weight-stable}
  Let $\Conf$ be a \RelC\ configuration and let $\MarMatchSet$ be the set of all matchings $\MarMatch$ of $\MarketConf{\Conf}$ such that $\PLeast{\Conf}{\MarMatch}$ holds.
  Assume that there exists at least one stable outcome, denoted $\MarOutcomeS$, for $\MarketConf{\Conf}$ such that $\MarMatch$ belongs to $\MarMatchSet$.
  Then, a matching $\MarMatchOpt$ that belongs to $\MarMatchSet$ is compatible with the stable payoff $\MarPayoffS$ if and only if $\WMarM{\Conf}{\MarMatchOpt} = \max_{\MarMatchp \in \MarMatchSet} \WMarM{\Conf}{\MarMatchp}$.
\end{lemma}

\newcommand{\POpt}[1]{P_{\Id{opt}}(#1)}

For any \RelC\ configuration $\Conf$, we define the predicate $\POpt{\Conf}$ to hold if for each man-optimal matching $\MarMatch$ of $\MarketConf{\Conf}$, $\PLeast{\Conf}{\MarMatch}$ holds.
Remark: It is easy to see that $\PLeast{\ConfFirst}{\MarMatch}$ holds for any matching $\MarMatch$ of $\MarketConf{\ConfFirst}$, and thus $\POpt{\ConfFirst}$ holds.
For a \RelC\ configuration $\Conf$ satisfying $\POpt{\Conf}$, the following lemma characterizes the man-optimal matchings of $\MarketConf{\Conf}$.

\begin{lemma}
  \label{lem:man-optimal-max-weight}
  Let $\Conf$ be a \RelC\ configuration such that $\POpt{\Conf}$ holds.
  Let $\MarMatchSet$ be the set of all matchings $\MarMatch$ of $\MarketConf{\Conf}$ such that $\PLeast{\Conf}{\MarMatch}$ holds.
  Then, a matching $\MarMatch$ is a man-optimal matching of $\MarketConf{\Conf}$ if and only if $\MarMatch$ belongs to $\MarMatchSet$ and $\WMarM{\Conf}{\MarMatch} = \max_{\MarMatchp \in \MarMatchSet} \WMarM{\Conf}{\MarMatchp}$.
\end{lemma}

\begin{proof}
  Let $\MarMatchSetOpt$ denote the set $\SetBuild{\MarMatch}{\MarMatch \in \MarMatchSet \land \WMarM{\Conf}{\MarMatch} = \max_{\MarMatchp \in \MarMatchSet} \WMarM{\Conf}{\MarMatchp}}$.
  Since $\POpt{\Conf}$ implies that all man-optimal matchings of $\MarketConf{\Conf}$ are included in $\MarMatchSet$, and since there is at least one man-optimal matching~\cite[Property 2]{DG85}, we deduce that $\MarMatchSet$ contains a man-optimal, and hence stable, matching.
  Thus, Lemma~\ref{lem:max-weight-stable} implies that the set of stable matchings of $\MarketConf{\Conf}$ is equal to $\MarMatchSetOpt$, and that each matching in $\MarMatchSetOpt$ is compatible with the man-optimal payoff in $\MarketConf{\Conf}$.
  \QEDWrap
\end{proof}

\begin{lemma}
  \label{lem:man-optimal-greedy-mwm}
  Let $\Conf = (\Auc, \IAuc)$ be a \RelC\ configuration such that $\POpt{\Conf}$ holds.
  Then, $\toUAPFunc{\Conf}$ is a bijection from the set of man-optimal matchings of $\MarketConf{\Conf}$ to the set of greedy MWMs of $\Auc$.
\end{lemma}

\begin{proof}
  Let $\MarMatchSet$ be the set of all matchings $\MarketConf{\Conf}$ of $\Market$ such that $\PLeast{\Conf}{\MarMatch}$ holds.
  Let $\MarMatchSetOpt$ denote the set $\SetBuild{\MarMatch}{\MarMatch \in \MarMatchSet \land \WMarM{\Conf}{\MarMatch} = \max_{\MarMatchp \in \MarMatchSet} \WMarM{\Conf}{\MarMatchp}}$.
  Lemma~\ref{lem:man-optimal-max-weight} implies that the set of man-optimal matchings of $\MarketConf{\Conf}$ is equal to $\MarMatchSetOpt$.
  Then the claim follows from Lemma~\ref{lem:weight-greedy-mwm}.
  \QEDWrap
\end{proof}

Having established the correspondence between the man-optimal matchings and the greedy MWMs given a \RelC\ configuration $\Conf$ such that $\POpt{\Conf}$ holds, we now show inductively in Lemma~\ref{lem:popt} that $\POpt{\Conf}$ holds for all \RelC\ configurations $\Conf$.
We start with two lemmas that are useful in proving Lemma~\ref{lem:popt};
the second one (Lemma~\ref{lem:man-optimal-monotone-reserve}) is also used in the proof of Theorem~\ref{thm:equivalence}.

\begin{lemma}
  \label{lem:man-optimal-ready}
  Let $\Conf$ be a \RelC\ configuration such that $\POpt{\Conf}$ holds, let $\MenResp$ denote $\ConfRes{\Conf}$, and let $\MarPayoffMenOpt$ denote the man-optimal payoff in $\MarketConf{\Conf}$.
  Then, for each bidder in $\CReady{\Conf}$, we have $\ManUtilOpt{\Man} = \ManResp{\Man}$, where $\Man$ denotes the man associated with that bidder.
\end{lemma}

\begin{proof}
  Let $\Conf$ be $(\Auc, \IAuc)$.
  If a bidder belongs to $\CReady{\Conf}$, then it is not matched in any greedy MWM of $\Auc$, and the definition of $\toUAPFunc{\Conf}$ and Lemma~\ref{lem:man-optimal-greedy-mwm} imply that it is not matched in any man-optimal matching of $\MarketConf{\Conf}$.
  Then the claim of the lemma follows from the stability of $\MarPayoffMenOpt$.
  \QEDWrap
\end{proof}

\begin{lemma}
  \label{lem:man-optimal-monotone-reserve}
  Let $\MenResp$ and $\MenRespp$ be two reserve utility vectors of the men such that $\MenResp \geq \MenRespp \geq \MenRes$.
  Let $\MarPayoffSp$ and $\MarPayoffSpp$ be the man-optimal payoffs of $\MarketRes{\MenResp}$ and $\MarketRes{\MenRespp}$, respectively.
  Then $\MenUtilp \geq \MenUtilpp$.
\end{lemma}

\begin{proof}
  The claim of the lemma follows directly from \cite[Property~3]{DG85}.
  \QEDWrap
\end{proof}

\begin{lemma}
  \label{lem:popt}
  Let $\Conf$ be a \RelC\ configuration such that for each \RelC\ configuration $\Confp$ that precedes $\Conf$, $\POpt{\Confp}$ holds.
  Then $\POpt{\Conf}$ holds.
\end{lemma}

\begin{proof}
  Consider an arbitrary \CanE\ execution that produces $\Conf$ at some iteration and let $(\Auc, \IAuc)$ be the configuration $\Conf$.
  For the sake of contradiction, assume $\POpt{\Conf}$ does not hold.
  Let $\MarOutcomeMenOpt{\MarMatch}$ be a man-optimal outcome for $\MarketConf{\Conf}$ such that $\PLeast{\Conf}{\MarMatch}$ does not hold and let $(\Man, \Woman)$ be a man-woman pair matched in $\MarMatch$ such that $\Woman$ does not belong to $\LeastTier{\Conf}{\Man}$.
  Let $\Womanp$ be an arbitrary element of $\LeastTier{\Conf}{\Man}$.
  Since Lemma~\ref{lem:reveal-reserve} implies that $\toBidder{\Man}{\Woman}$ belongs to $\Auc$, we deduce that $\ManPref{\Man}{\Woman} > \ManPref{\Man}{\Womanp}$.
  Let $\Confp$ be the configuration at the beginning of the iteration that reveals the tier of $\Man$ corresponding to $\LeastTier{\Conf}{\Man}$, i.e., the iteration at which $\toBidder{\Man}{\Womanp}$ is added at line~\ref{line:reveal} of Alg.~\ref{alg:to-uap}.
  Let $\Womanpp$ be an arbitrary element of $\LeastTier{\Confp}{\Man}$, let $\MenResp$ denote $\ConfRes{\Confp}$, and let $\MarPayoffMenOptp$ denote the man-optimal payoff in $\MarketConf{\Confp}$.
  Then, we deduce that $\ManPref{\Man}{\Woman} \geq \ManPref{\Man}{\Womanpp} > \ManPref{\Man}{\Womanp} \geq \ManPrefUnmatched{\Man}$, where the last inequality follows from Lemma~\ref{lem:reveals-acceptable}.
  Thus, $\PRevealed{\Confp}{\Man}$ does not hold, and hence $\ManResp{\Man} = \ResCS \lexp{\ManPref{\Man}{\Womanpp}}$.
  Since $\toBidder{\Man}{\Womanp}$ belongs to $\CReady{\Confp}$, and since $\POpt{\Confp}$ holds, Lemma~\ref{lem:man-optimal-ready} implies that $\ManUtilOptp{\Man} = \ManResp{\Man} = \ResCS \lexp{\ManPref{\Man}{\Womanpp}}$.
  On the other hand, since the man-woman pair $(\Man, \Woman)$ is matched in $\MarMatch$, which is a man-optimal matching of $\MarketConf{\Conf}$, Lemma~\ref{lem:man-lower-bound} implies that $\ManUtilOpt{\Man} \geq \lexp{\ManPref{\Man}{\Woman}}$.
  Combining the results of the preceding two sentences, we conclude that $\ManUtilOpt{\Man} \geq \lexp{\ManPref{\Man}{\Woman}} > \ResCS \lexp{\ManPref{\Man}{\Womanpp}} = \ManUtilOptp{\Man}$, contradicting Lemma~\ref{lem:man-optimal-monotone-reserve} since $\ConfRes{\Conf} \leq \ConfRes{\Confp}$.
  \QEDWrap
\end{proof}

We are now ready to establish our equivalence result in Theorem~\ref{thm:equivalence}.
We start with a useful lemma.

\begin{lemma}
  \label{lem:every-man-matched}
  For each greedy MWM $\Match$ of $\UAP{\IAuc}$ and each man $\Man$, some bidder associated with $\Man$ is matched in $\Match$.
\end{lemma}

\begin{proof}
  Suppose the claim does not hold, and let $\Match$ be a greedy MWM of $\UAP{\IAuc}$ and let $\Man$ be a man such that no bidder associated with $\Man$ is matched in $\Match$.
  Then all of the bidders associated with $\Man$ belong to $\UAP{\IAuc}$, for otherwise $\CReady{\ConfFinal}$ is nonempty, contradicting the definition of $\ConfFinal$ (recall that $\ConfFinal = (\UAP{\IAuc}, \IAuc)$ is the unique final configuration of any \CanE\ execution).
  Thus, we deduce that $\PRevealed{\ConfFinal}{\Man}$ holds, and Lemma~\ref{lem:all-revealed-matched} implies that some bidder associated with $\Man$ is matched in $\Match$, a contradiction.
  \QEDWrap
\end{proof}

\begin{theorem}
  \label{thm:equivalence}
  $\toUAPFunc{\ConfFinal}$ is a bijection from the set of man-optimal matchings of $\Market$ to the set of greedy MWMs of $\UAP{\IAuc}$.
\end{theorem}

\begin{proof}
  Observe that $\POpt{\ConfFinal}$ holds by repeated application of Lemma~\ref{lem:popt}.
  Let $\MarPayoffMenOpt$ denote the man-optimal payoff in $\MarketConf{\ConfFinal}$.
  Let $\MenResp$ denote $\ConfRes{\ConfFinal}$ and recall that $\MenRes$ is the reserve utility vector of the men in $\Market$.
  We now prove two useful claims.

  Claim 1: Any man-optimal outcome $\MarOutcomeMenOpt{\MarMatch}$ for $\MarketConf{\ConfFinal}$ is stable for $\Market$.
  Let $\MarOutcomeMenOpt{\MarMatch}$ be a man-optimal outcome for $\MarketConf{\ConfFinal}$.
  Since $\MarketConf{\ConfFinal}$ and $\Market$ differ only in the reserve utility vectors of the men, it is enough to show that for each man $\Man$ who is unmatched in $\MarMatch$, we have $\ManUtilOpt{\Man} = \ManRes{\Man}$.
  Let $\Man$ be a man who is unmatched in $\MarMatch$.
  Observe that $\PRevealed{\ConfFinal}{\Man}$ holds, for otherwise no bidder associated with $\Man$ is matched in $\toUAP{\ConfFinal}{\MarMatch}$, which is a greedy MWM of $\UAP{\IAuc}$ by Lemma~\ref{lem:man-optimal-greedy-mwm}, contradicting Lemma~\ref{lem:every-man-matched}.
  Thus $\ManResp{\Man} = \ManPriS \lexp{\ManPref{\Man}{\Unmatched}} = \ManRes{\Man}$.
  Then, by the stability of $\MarOutcomeMenOpt{\MarMatch}$ for $\MarketConf{\ConfFinal}$, we conclude that $\ManUtilOpt{\Man} = \ManResp{\Man} = \ManRes{\Man}$.

  Claim 2: Any man-optimal outcome for $\Market$ is stable for $\MarketConf{\ConfFinal}$.
  Claim 1 and Lemma~\ref{lem:man-optimal-monotone-reserve} imply that the payoff $\MarPayoffMenOpt$, which is the man-optimal payoff in $\MarketConf{\ConfFinal}$, is also the man-optimal payoff in $\Market$.
  Let $\MarOutcomeMenOpt{\MarMatch}$ be a man-optimal outcome for $\Market$.
  As in the proof of Claim 1, since $\MarketConf{\ConfFinal}$ and $\Market$ differ only in the reserve utility vectors of the men, it is enough to show that for each man $\Man$ who is unmatched in $\MarMatch$, we have $\ManUtilOpt{\Man} = \ManResp{\Man}$.
  Let $\Man$ be a man who is unmatched in $\MarMatch$.
  By the stability of $\MarOutcomeMenOpt{\MarMatch}$ for $\Market$, we deduce that $\ManUtilOpt{\Man} = \ManRes{\Man}$.
  By the individual rationality of $\MarPayoffMenOpt$ for $\MarketConf{\ConfFinal}$, we deduce that $\ManUtilOpt{\Man} \geq \ManResp{\Man}$.
  Since $\MenResp \geq \MenRes$, we conclude that $\ManUtilOpt{\Man} = \ManRes{\Man} = \ManResp{\Man}$.

  Claims 1 and 2, and Lemma~\ref{lem:man-optimal-monotone-reserve} imply that the set of man-optimal matchings of $\Market$ is equal to the set of man-optimal matchings of $\MarketConf{\ConfFinal}$.
  Then the theorem follows from Lemma~\ref{lem:man-optimal-greedy-mwm} since $\POpt{\ConfFinal}$ holds.
  \QEDWrap
\end{proof}

\begin{proof}[\ConfOrFull{}{Proof of }Theorem~\ref{thm:implementation}]
  We construct a \TSM\ $\Market$ associated with the instance of the stable marriage market with indifferences, and we run Alg.~\ref{alg:smiw}, which admits an $O(n^4)$-time implementation, by setting the edge weights as described in App.~\ref{app:iuap-weights}.
  Then the result follows from Theorems~\ref{thm:strategyproofness} and \ref{thm:equivalence}.
  \QEDWrap
\end{proof}

\subsection{Proof of Lemma~\ref{lem:max-weight-stable}}
\label{sec:equivalence-proof}

The purpose of this appendix is to prove Lemma~\ref{lem:max-weight-stable}.
We start with some useful lemmas.

\begin{lemma}
  \label{lem:women-payment}
  Let $\Conf$ be a \RelC\ configuration, let $\MarMatch$ be a matching of $\MarketConf{\Conf}$, and let $\WomenUtil$ be a utility vector of the women such that $\WomanUtilS = \WomanResS$ for each woman $\Woman$ who is unmatched in $\MarMatch$.
  Then,
  \begin{equation*}
     \sum_{\EachWoman} \WomanUtilS - \WMarM{\Conf}{\MarMatch} + \MenPriUnmatchedS{\MarMatch} = \sum_{\IsMatchedS{\Woman}} \WomanCompMatchS{\MarMatch} .
  \end{equation*}
\end{lemma}

\begin{proof}
  \begin{align*}
    \sum_{\EachWoman} \WomanUtilS - \WMarM{\Conf}{\MarMatch} + \MenPriUnmatchedS{\MarMatch}
    &= \sum_{\EachWoman} \WomanUtilS - \NN \cdot \WMPref{\MarMatch} - \sum_{\IsMatchedS{\Man}} \ManPriS \\
    &= \sum_{\IsMatchedS{\Woman}} \WomanUtilS - \NN \sum_{\IsMatchedS{\Woman}} \WomanPref{\MOf{\MarMatch}{\Woman}}{\Woman} - \sum_{\IsMatchedS{\Woman}} \ManPri{\MOf{\MarMatch}{\Woman}} \\
    &\quad + \sum_{\IsUnmatchedS{\Woman}} \WomanResS - \NN \sum_{\IsUnmatchedS{\Woman}} \WomanPrefUnmatchedS \\
    &= \sum_{\IsMatchedS{\Woman}} \pr{\WomanUtilS - \WomanPref{\MOf{\MarMatch}{\Woman}}{\Woman} \NN - \ManPri{\MOf{\MarMatch}{\Woman}}} \\
    &= \sum_{\IsMatchedS{\Woman}} \WomanCompMatchS{\MarMatch} ,
  \end{align*}
  where the third equality follows since $\WomanResS = \NN \WomanPrefUnmatchedS$.
  \QEDWrap
\end{proof}

\begin{lemma}
  \label{lem:stable-weight}
  Let $\Conf$ be a \RelC\ configuration and let $\MarOutcomeS$ be a stable outcome for $\MarketConf{\Conf}$.
  Then,
  \begin{equation*}
    \WMarM{\Conf}{\MarMatch} = \sum_{\IsMatchedS{\Man}} \ManCompMatchS{\MarMatch} + \sum_{\EachWoman} \WomanUtilS + \MenPriUnmatchedS{\MarMatch} .
  \end{equation*}
\end{lemma}

\begin{proof}
  The stability and feasibility of $\MarOutcomeS$ imply that $\ManCompS + \WomanCompS = 0$ for each man-woman pair $(\Man, \Woman)$ matched in $\MarMatch$.
  Thus, $\sum_{\IsMatchedS{\Man}} \ManCompMatchS{\MarMatch} + \sum_{\IsMatchedS{\Woman}} \WomanCompMatchS{\MarMatch} = 0$, and the claim follows from Lemma~\ref{lem:women-payment}, since the stability of $\MarOutcomeS$ implies that $\WomanUtilS = \WomanResS$ for each woman $\Woman$ who is unmatched in $\MarMatch$.
  \QEDWrap
\end{proof}

\begin{lemma}
  \label{lem:u-equal}
  Let $\Conf$ be a \RelC\ configuration, let $\MarMatch$ and $\MarMatchp$ be two matchings such that $\PLeast{\Conf}{\MarMatch}$ and $\PLeast{\Conf}{\MarMatchp}$ hold, let $\MenResp$ denote $\ConfRes{\Conf}$, and let $\MenUtil$ be a utility vector such that for each man $\Man$, $\ManUtil{\Man} = \ManResp{\Man}$ if $\IsUnmatched{\Man}{\MarMatch}$ or $\IsUnmatched{\Man}{\MarMatchp}$.
  Then $\sum_{\IsMatched{\Man}{\MarMatch}} \ManCompMatchS{\MarMatch} + \MenPriUnmatchedS{\MarMatch} = \sum_{\IsMatched{\Man}{\MarMatchp}} \ManCompMatchS{\MarMatchp} + \MenPriUnmatchedS{\MarMatchp}$.
\end{lemma}

\newcommand{\ManCompUnmatchedS} {\ManComp{\Man}{\Unmatched}{\MenUtil}}

\begin{proof}
  Let $\ManCompUnmatchedS$ denote $0$.
  We show that $\ManCompMatchS{\MarMatch} + \ManPriUnmatchedS{\MarMatch} = \ManCompMatchS{\MarMatchp} + \ManPriUnmatchedS{\MarMatchp}$ for each man $\Man$.
  Let $\Man$ be an arbitrary man.
  We consider six cases.

  Case 1: $\IsMatched{\Man}{\MarMatch}$ and $\IsMatched{\Man}{\MarMatchp}$.
  Then $\ManPriUnmatchedS{\MarMatch} = \ManPriUnmatchedS{\MarMatchp} = 0$.
  $\PLeast{\Conf}{\MarMatch}$ and $\PLeast{\Conf}{\MarMatchp}$ and imply that $\ManPref{\Man}{\MOf{\MarMatch}{\Man}} = \ManPref{\Man}{\MOf{\MarMatchp}{\Man}}$, and hence $\ManCompMatchS{\MarMatch} = \ManCompMatchS{\MarMatchp}$.

  Case 2: $\MOf{\MarMatch}{\Man} = \MOf{\MarMatchp}{\Man} = \Unmatched$.
  In this case $\ManPriUnmatchedS{\MarMatch}$ (resp., $\ManPriUnmatchedS{\MarMatchp}$) is independent of $\MarMatch$ (resp., $\MarMatchp$).

  Case 3 (resp., case 4): $\IsMatched{\Man}{\MarMatch}$ and $\IsUnmatched{\Man}{\MarMatchp}$ (resp., $\IsUnmatched{\Man}{\MarMatch}$ and $\IsMatched{\Man}{\MarMatchp}$) and $\PRevealed{\Conf}{\Man}$.
  Since $\PRevealed{\Conf}{\Man}$ holds, we deduce that $\ManResp{\Man} = \ManPriS \lexp{\ManPref{\Man}{\Unmatched}}$.
  Then, since $\ManUtil{\Man} = \ManResp{\Man}$, $\PLeast{\Conf}{\MarMatch}$ (resp., $\PLeast{\Conf}{\MarMatchp}$) implies that $\ManCompMatchS{\MarMatch} + \ManPriUnmatchedS{\MarMatch} = \ManComp{\Man}{\MOf{\MarMatch}{\Man}}{\MenResp} = \ManPriS$ (resp., $\ManCompMatchS{\MarMatchp} + \ManPriUnmatchedS{\MarMatchp} = \ManComp{\Man}{\MOf{\MarMatchp}{\Man}}{\MenResp} = \ManPriS$).
  Since $\IsUnmatched{\Man}{\MarMatchp}$ (resp., $\IsUnmatched{\Man}{\MarMatch}$), we deduce that $\ManCompMatchS{\MarMatchp} + \ManPriUnmatchedS{\MarMatchp}$ (resp., $\ManCompMatchS{\MarMatch} + \ManPriUnmatchedS{\MarMatch}$) is equal to $0 + \ManPriS = \ManPriS$.

  Case 5 (resp., case 6): $\IsMatched{\Man}{\MarMatch}$ and $\IsUnmatched{\Man}{\MarMatchp}$ (resp., $\IsUnmatched{\Man}{\MarMatch}$ and $\IsMatched{\Man}{\MarMatchp}$) and $\lnot \PRevealed{\Conf}{\Man}$.
  Let $\Woman$ denote $\MOf{\MarMatch}{\Man}$ (resp., $\MOf{\MarMatchp}{\Man}$).
  Since $\PLeast{\Conf}{\MarMatch}$ (resp., $\PLeast{\Conf}{\MarMatchp}$) and $\lnot \PRevealed{\Conf}{\Man}$, we deduce that $\ManResp{\Man} = \ResCS \lexp{\ManPref{\Man}{\Woman}}$.
  Then, since $\ManUtil{\Man} = \ManResp{\Man}$, we deduce that $\ManCompMatchS{\MarMatch} + \ManPriUnmatchedS{\MarMatch} = \ManComp{\Man}{\MOf{\MarMatch}{\Man}}{\MenResp} = \ResCS$ (resp., $\ManCompMatchS{\MarMatchp} + \ManPriUnmatchedS{\MarMatchp} = \ManComp{\Man}{\MOf{\MarMatchp}{\Man}}{\MenResp} = \ResCS$).
  Since $\IsUnmatched{\Man}{\MarMatchp}$ (resp., $\IsUnmatched{\Man}{\MarMatch}$), we deduce that $\ManCompMatchS{\MarMatchp} + \ManPriUnmatchedS{\MarMatchp}$ (resp., $\ManCompMatchS{\MarMatch} + \ManPriUnmatchedS{\MarMatch}$) is equal to $0 + \ResCS = \ResCS$.
  \QEDWrap
\end{proof}

\newcommand{\ClaimMaxWeightUnmatchedMen}{(1)}
\newcommand{\ClaimMaxWeightUnmatchedWomen}{(2)}
\newcommand{\ClaimMaxWeightUEqual}{(3)}

\begin{lemma}
  \label{lem:max-weight-compatible}
  Let $\Conf$ be a \RelC\ configuration and let $\MenResp$ denote $\ConfRes{\Conf}$.
  Let $\MarMatchSet$ be the set of all matchings $\MarMatch$ of $\MarketConf{\Conf}$ such that $\PLeast{\Conf}{\MarMatch}$ holds.
  Let $\MarOutcomeS$ be a stable outcome for $\MarketConf{\Conf}$ such that $\MarMatch$ belongs to $\MarMatchSet$.
  Let $\MarMatchOpt$ be a matching in $\MarMatchSet$ such that $\WMarM{\Conf}{\MarMatchOpt} = \max_{\MarMatchp \in \MarMatchSet} \WMarM{\Conf}{\MarMatchp}$.
  Then the following claims hold:
  \ClaimMaxWeightUnmatchedMen\ $\ManUtil{\Man} = \ManResp{\Man}$ if $\IsUnmatched{\Man}{\MarMatchOpt}$;
  \ClaimMaxWeightUnmatchedWomen\ $\WomanUtil{\Woman} = \WomanRes{\Woman}$ if $\IsUnmatched{\Woman}{\MarMatchOpt}$;
  \ClaimMaxWeightUEqual\ $\sum_{\IsMatchedS{\Man}} \ManCompMatchS{\MarMatch} + \MenPriUnmatchedS{\MarMatch} = \sum_{\IsMatched{\Man}{\MarMatchOpt}} \ManCompMatchS{\MarMatchOpt} + \MenPriUnmatchedS{\MarMatchOpt}$.
\end{lemma}

\begin{proof}
  Let $\PathCycleColl$ denote the symmetric difference of $\MarMatch$ and $\MarMatchOpt$.
  It is easy to see that $\PathCycleColl$ is a collection of positive length paths and cycles.
  In order to prove Claim~\ClaimMaxWeightUnmatchedMen\ of the lemma, consider an arbitrary man $\Man$ such that $\IsUnmatched{\Man}{\MarMatchOpt}$.
  If $\IsUnmatchedS{\Man}$, then the stability of $\MarPayoffS$ establishes the claim, so assume that $\IsMatchedS{\Man}$.
  Then, $\Man$ is an endpoint of a path in $\PathCycleColl$;
  let $\Path$ denote this path.
  The edges of $\Path$ alternate between edges that are matched in $\MarMatch$ and edges that are matched in $\MarMatchOpt$.
  We consider two cases.

  \newcommand{\pl}{k}

  Case 1: The other endpoint of $\Path$ is a man $\Manp$ such that $\IsUnmatched{\Manp}{\MarMatch}$ and $\IsMatched{\Man}{\MarMatchOpt}$.
  Let $\Path$ be $\Seq{\Man = \ManSub{1}, \WomanSub{1}, \dotsc, \WomanSub{\pl}, \ManSub{\pl+1} = \Manp}$ for some $\pl \geq 1$.
  Then, since $\MOf{\MarMatch}{\ManSub{\ell}} = \WomanSub{\ell}$ for $1 \leq \ell \leq \pl$, the stability of $\MarOutcomeS$ implies that
  \begin{equation*}
    \sum_{1 \leq \ell \leq \pl} \sbr{\ManComp{\ManSub{\ell}}{\WomanSub{\ell}}{\MenUtil} + \WomanComp{\WomanSub{\ell}}{\ManSub{\ell}}{\WomenUtil}} = 0 .
  \end{equation*}
  The stability of $\MarOutcomeS$ also implies that
  \begin{equation*}
    \sum_{1 \leq \ell \leq \pl} \sbr{\ManComp{\ManSub{\ell+1}}{\WomanSub{\ell}}{\MenUtil} + \WomanComp{\WomanSub{\ell}}{\ManSub{\ell+1}}{\WomenUtil}} \geq 0 .
  \end{equation*}
  By subtracting the latter equation from the former, we obtain the following, since $\PLeast{\Conf}{\MarMatch}$ and $\PLeast{\Conf}{\MarMatchOpt}$ imply $\ManComp{\ManSub{\ell}}{\WomanSub{\ell}}{\MenUtil} = \ManComp{\ManSub{\ell}}{\WomanSub{\ell-1}}{\MenUtil}$ for $1 < \ell \leq \pl$:
  \begin{align}
    0
    &\geq \ManComp{\Man}{\WomanSub{1}}{\MenUtil} - \ManComp{\Manp}{\WomanSub{\pl}}{\MenUtil} + \sum_{1 \leq \ell \leq \pl} \sbr{\WomanComp{\WomanSub{\ell}}{\ManSub{\ell}}{\WomenUtil} - \WomanComp{\WomanSub{\ell}}{\ManSub{\ell+1}}{\WomenUtil}} \nonumber \\
    &= \pr{\ManComp{\Man}{\WomanSub{1}}{\MenUtil} - \ManPri{\Man}} - \pr{\ManComp{\Manp}{\WomanSub{\pl}}{\MenUtil} - \ManPri{\Manp}} + \NN \sum_{1 \leq \ell \leq \pl} \pr{\WomanPref{\ManSub{\ell+1}}{\WomanSub{\ell}} - \WomanPref{\ManSub{\ell}}{\WomanSub{\ell}}} . \label{eq:alt-man-1}
  \end{align}
  Observe that
  \begin{equation}
    \ManComp{\Man}{\WomanSub{1}}{\MenUtil} - \ManPri{\Man} \geq \ManComp{\Man}{\WomanSub{1}}{\MenResp} - \ManPri{\Man} \geq \ResC{\Man} - \ManPri{\Man} > -\NN ,\label{eq:obs-man-1}
  \end{equation}
  since the individual rationality of $\MarPayoffS$ implies $\ManUtil{\Man} \geq \ManResp{\Man}$ and since $\PLeast{\Conf}{\MarMatch}$ holds.
  Also observe that
  \begin{equation*}
    \ManComp{\Manp}{\WomanSub{\pl}}{\MenUtil} - \ManPri{\Manp} = \ManComp{\Manp}{\WomanSub{\pl}}{\MenResp} - \ManPri{\Manp} \leq 0 ,
  \end{equation*}
  since the stability of $\MarOutcomeS$ implies $\ManUtil{\Manp} = \ManResp{\Manp}$ and since $\PLeast{\Conf}{\MarMatchOpt}$ holds.
  These two observations imply that the third term in~\eqref{eq:alt-man-1} is nonpositive, for otherwise it would be at least $\NN$ (since all $\WomenPref$ values are integers), violating the inequality.
  The third term in~\eqref{eq:alt-man-1} is nonnegative, for otherwise $\MarMatchOpt$ could be augmented along $\Path$ to yield another matching $\MarMatchp$ such that $\PLeast{\Conf}{\MarMatchp}$ and $\WMarM{\Conf}{\MarMatchp} - \WMarM{\Conf}{\MarMatchOpt} \geq \NN - \ManPri{\Manp} > 0$ (since all $\WomenPref$ values are integers and $\NN > \ManPri{\Manp}$), contradicting the definition of $\MarMatchOpt$.
  Thus, we may rewrite inequality~(\ref{eq:alt-man-1}) as
  \begin{equation}
    \ManComp{\Manp}{\WomanSub{\pl}}{\MenResp} - \ManPri{\Manp} \geq \ManComp{\Man}{\WomanSub{1}}{\MenUtil} - \ManPri{\Man} .\label{eq:alt-man-1-rw}
  \end{equation}
  We consider the following four subcases.

  Case 1.1: $\PRevealed{\Conf}{\Man}$ and $\PRevealed{\Conf}{\Manp}$.
  Since $\PLeast{\Conf}{\MarMatchOpt}$ and $\PRevealed{\Conf}{\Manp}$ hold, we deduce that $\ManResp{\Manp} = \ManPri{\Manp} \lexp{\ManPref{\Manp}{\WomanSub{\pl}}}$, and hence that LHS of~(\ref{eq:alt-man-1-rw}) is $0$.
  Since $\PLeast{\Conf}{\MarMatch}$ and $\PRevealed{\Conf}{\Man}$ hold, we deduce that $\ManResp{\Man} = \ManPri{\Man} \lexp{\ManPref{\Manp}{\WomanSub{1}}}$, and hence that RHS of~(\ref{eq:alt-man-1-rw}) is at least $0$ by the first inequality of~(\ref{eq:obs-man-1}).
  Thus, we conclude that $\ManComp{\Man}{\WomanSub{1}}{\MenUtil} = \ManPri{\Man}$, which implies $\ManUtil{\Man} = \ManResp{\Man}$.

  Case 1.2: $\lnot \PRevealed{\Conf}{\Man}$ and $\lnot \PRevealed{\Conf}{\Manp}$.
  First observe that $\ManPri{\Manp} > \ManPri{\Man}$, for otherwise $\MarMatchOpt$ could be augmented along $\Path$ to yield another matching $\MarMatchp$ such that $\PLeast{\Conf}{\MarMatchp}$ and $\WMarM{\Conf}{\MarMatchp} - \WMarM{\Conf}{\MarMatchOpt} = \ManPri{\Man} - \ManPri{\Manp} > 0$, contradicting the definition of $\MarMatchOpt$.
  Since $\PLeast{\Conf}{\MarMatchOpt}$ holds and $\PRevealed{\Conf}{\Manp}$ does not hold, we deduce that $\ManResp{\Manp} = \ResC{\Manp} \lexp{\ManPref{\Manp}{\WomanSub{\pl}}}$, and hence that LHS of~(\ref{eq:alt-man-1-rw}) is $\ResC{\Manp} - \ManPri{\Manp}$.
  However, the RHS of~(\ref{eq:alt-man-1-rw}) is at least $\ResC{\Man} - \ManPri{\Man}$ by~(\ref{eq:obs-man-1}), contradicting the inequality $\ManPri{\Manp} > \ManPri{\Man}$.

  Case 1.3: $\lnot \PRevealed{\Conf}{\Man}$ and $\PRevealed{\Conf}{\Manp}$.
  This case is not possible, for otherwise $\MarMatchOpt$ could be augmented along $\Path$ to yield another matching $\MarMatchp$ such that $\PLeast{\Conf}{\MarMatchp}$ and $\WMarM{\Conf}{\MarMatchp} - \WMarM{\Conf}{\MarMatchOpt} = \ManPri{\Man} - \ResC{\Man} > 0$ (since $\ManPri{\Man} > \ResC{\Man}$ for each man $\Man$), contradicting the definition of $\MarMatchOpt$.

  Case 1.4: $\PRevealed{\Conf}{\Man}$ and $\lnot \PRevealed{\Conf}{\Manp}$.
  Since $\PLeast{\Conf}{\MarMatchOpt}$ holds and $\PRevealed{\Conf}{\Manp}$ does not hold, we deduce that $\ManResp{\Manp} = \ResC{\Manp} \lexp{\ManPref{\Manp}{\WomanSub{\pl}}}$, and hence that LHS of~(\ref{eq:alt-man-1-rw}) is $\ResC{\Manp} - \ManPri{\Manp} < 0$.
  Since $\PLeast{\Conf}{\MarMatch}$ and $\PRevealed{\Conf}{\Man}$ hold, we deduce that $\ManResp{\Man} = \ManPri{\Man} \lexp{\ManPref{\Manp}{\WomanSub{1}}}$, and hence that RHS of~(\ref{eq:alt-man-1-rw}) is at least $0$ by the first inequality of~(\ref{eq:obs-man-1}), a contradiction.

  Case 2: The other endpoint of $\Path$ is a woman $\Woman$ such that $\IsMatchedS{\Woman}$ and $\IsUnmatched{\Woman}{\MarMatchOpt}$.
  Let $\Path$ be $\Seq{\Man = \ManSub{1}, \WomanSub{1}, \dotsc, \ManSub{\pl}, \WomanSub{\pl} = \Woman}$ for some $\pl \geq 1$.
  Then, since $\MOf{\MarMatch}{\ManSub{\ell}} = \WomanSub{\ell}$ for $1 \leq \ell \leq \pl$, the stability of $\MarOutcomeS$ implies that
  \begin{equation*}
    \sum_{1 \leq \ell \leq \pl} \sbr{\ManComp{\ManSub{\ell}}{\WomanSub{\ell}}{\MenUtil} + \WomanComp{\WomanSub{\ell}}{\ManSub{\ell}}{\WomenUtil}} = 0 .
  \end{equation*}
  The stability of $\MarOutcomeS$ also implies that
  \begin{equation*}
    \sum_{1 < \ell \leq \pl} \sbr{\ManComp{\ManSub{\ell}}{\WomanSub{\ell-1}}{\MenUtil} + \WomanComp{\WomanSub{\ell-1}}{\ManSub{\ell}}{\WomenUtil}} \geq 0 .
  \end{equation*}
  By subtracting the latter equation from the former, we obtain the following, since $\PLeast{\Conf}{\MarMatch}$ and $\PLeast{\Conf}{\MarMatchOpt}$ imply $\ManComp{\ManSub{\ell}}{\WomanSub{\ell}}{\MenUtil} = \ManComp{\ManSub{\ell}}{\WomanSub{\ell-1}}{\MenUtil}$ for $1 < \ell \leq \pl$:
  \begin{align}
    0
    &\geq \ManComp{\Man}{\WomanSub{1}}{\MenUtil} + \sum_{1 \leq \ell \leq \pl} \WomanComp{\WomanSub{\ell}}{\ManSub{\ell}}{\WomenUtil} - \sum_{1 < \ell \leq \pl} \WomanComp{\WomanSub{\ell-1}}{\ManSub{\ell}}{\WomenUtil} \nonumber \\
    &= \pr{\ManComp{\Man}{\WomanSub{1}}{\MenUtil} - \ManPri{\Man}} + \WomanUtil{\Woman} + \NN \sbr{ \sum_{1 < \ell \leq \pl} \WomanPref{\ManSub{\ell}}{\WomanSub{\ell-1}} - \sum_{1 \leq \ell \leq \pl} \WomanPref{\ManSub{\ell}}{\WomanSub{\ell}}} \nonumber \\
    &= \pr{\ManComp{\Man}{\WomanSub{1}}{\MenUtil} - \ManPri{\Man}} + \pr{\WomanUtil{\Woman} - \WomanRes{\Woman}} + \NN \sbr{ \pr{\WomanPrefUnmatched{\Woman} + \sum_{1 < \ell \leq \pl} \WomanPref{\ManSub{\ell}}{\WomanSub{\ell-1}}} - \sum_{1 \leq \ell \leq \pl} \WomanPref{\ManSub{\ell}}{\WomanSub{\ell}}} . \label{eq:alt-man-2}
  \end{align}
  Observe that~(\ref{eq:obs-man-1}) holds for the same reasons pointed out in Case 1, and that the second term in~\eqref{eq:alt-man-2} is nonnegative by the individual rationality of $\MarPayoffS$.
  Moreover, the third term is nonnegative, for otherwise $\MarMatchOpt$ could be augmented along $\Path$ to yield another matching $\MarMatchp$ such that $\PLeast{\Conf}{\MarMatchp}$ and $\WMarM{\Conf}{\MarMatchp} - \WMarM{\Conf}{\MarMatchOpt} \geq \NN$ (since all $\WomenPref$ values are integers), contradicting the definition of $\MarMatchOpt$.
  We consider two subcases.

  Case 2.1: $\PRevealed{\Conf}{\Man}$.
  Then $\PLeast{\Conf}{\MarMatch}$ implies that $\ManResp{\Man} = \ManPri{\Man} \lexp{\ManPref{\Manp}{\WomanSub{1}}}$, and we deduce that the first term in~(\ref{eq:alt-man-2}) is at least $0$ by the first inequality of~(\ref{eq:obs-man-1}).
  Thus, we conclude that $\WomanUtil{\Woman} = \WomanRes{\Woman}$ and that $\ManComp{\Man}{\WomanSub{1}}{\MenUtil} = \ManPri{\Man}$, which implies $\ManUtil{\Man} = \ManResp{\Man}$.

  Case 2.2: $\lnot \PRevealed{\Conf}{\Man}$.
  In this case the third term in~(\ref{eq:alt-man-2}) is positive, and thus is at least $\NN$ (since all $\WomenPref$ values are integers), for otherwise (if it is $0$), $\MarMatchOpt$ could be augmented along $\Path$ to yield another matching $\MarMatchp$ such that $\PLeast{\Conf}{\MarMatchp}$ and $\WMarM{\Conf}{\MarMatchp} - \WMarM{\Conf}{\MarMatchOpt} = \ManPri{\Man} - \ResC{\Man} > 0$, contradicting the definition of $\MarMatchOpt$.
  Since (\ref{eq:obs-man-1}) implies that the first term of~(\ref{eq:alt-man-2}) is greater than $-\NN$, we deduce that the sum of the three terms in~(\ref{eq:alt-man-2}) is positive, a contradiction.

  In order to prove Claim~\ClaimMaxWeightUnmatchedWomen\ of the lemma, consider an arbitrary woman $\Woman$ such that $\IsUnmatched{\Woman}{\MarMatchOpt}$.
  If $\IsUnmatchedS{\Woman}$, then the stability of $\MarPayoffS$ establishes the claim, so assume that $\IsMatchedS{\Woman}$.
  Then, $\Woman$ is an endpoint of a path in $\PathCycleColl$;
  let $\Path$ denote this path.
  The edges of $\Path$ alternate between edges that are matched in $\MarMatch$ and edges that are matched in $\MarMatchOpt$.
  We consider two cases.

  Case 1: The other endpoint of $\Path$ is a woman $\Womanp$ such that $\IsUnmatched{\Womanp}{\MarMatch}$ and $\IsMatched{\Woman}{\MarMatchOpt}$.
  Let $\Path$ be $\Seq{\Woman = \WomanSub{1}, \ManSub{1}, \dotsc, \ManSub{\pl}, \WomanSub{\pl+1} = \Womanp}$ for some $\pl \geq 1$.
  Then, since $\MOf{\MarMatch}{\ManSub{\ell}} = \WomanSub{\ell}$ for $1 \leq \ell \leq \pl$, the stability of $\MarOutcomeS$ implies that
  \begin{equation*}
    \sum_{1 \leq \ell \leq \pl} \sbr{\ManComp{\ManSub{\ell}}{\WomanSub{\ell}}{\MenUtil} + \WomanComp{\WomanSub{\ell}}{\ManSub{\ell}}{\WomenUtil}} = 0 .
  \end{equation*}
  The stability of $\MarOutcomeS$ also implies that
  \begin{equation*}
    \sum_{1 \leq \ell \leq \pl} \sbr{\ManComp{\ManSub{\ell}}{\WomanSub{\ell+1}}{\MenUtil} + \WomanComp{\WomanSub{\ell+1}}{\ManSub{\ell}}{\WomenUtil}} \geq 0 .
  \end{equation*}
  By subtracting the latter equation from the former, we obtain the following, since $\PLeast{\Conf}{\MarMatch}$ and $\PLeast{\Conf}{\MarMatchOpt}$ imply $\ManComp{\ManSub{\ell}}{\WomanSub{\ell}}{\MenUtil} = \ManComp{\ManSub{\ell}}{\WomanSub{\ell+1}}{\MenUtil}$ for $1 \leq \ell \leq \pl$:
  \begin{align}
    0
    &\geq \sum_{1 \leq \ell \leq \pl} \sbr{\WomanComp{\WomanSub{\ell}}{\ManSub{\ell}}{\WomenUtil} - \WomanComp{\WomanSub{\ell+1}}{\ManSub{\ell}}{\WomenUtil}} \nonumber \\
    &= \WomanUtil{\Woman} - \WomanUtil{\Womanp} + \NN \sum_{1 \leq \ell \leq \pl} \pr{\WomanPref{\ManSub{\ell}}{\WomanSub{\ell+1}} - \WomanPref{\ManSub{\ell}}{\WomanSub{\ell}}} \nonumber \\
    &= \pr{\WomanUtil{\Woman} - \WomanRes{\Woman}} + \NN \sbr{ \pr{\WomanPrefUnmatched{\Woman} + \sum_{1 \leq \ell \leq \pl} \WomanPref{\ManSub{\ell}}{\WomanSub{\ell+1}}} - \pr{\WomanPrefUnmatched{\Womanp} + \sum_{1 \leq \ell \leq \pl} \WomanPref{\ManSub{\ell}}{\WomanSub{\ell}}}} \label{eq:alt-woman-1},
  \end{align}
  where the last equality follows since $\IsUnmatched{\Womanp}{\MarMatch}$ implies $\WomanUtil{\Womanp} = \WomanRes{\Womanp} = \WomanPrefUnmatched{\Womanp} \NN$.
  The first term in~\eqref{eq:alt-woman-1} is nonnegative by the individual rationality of $\MarPayoffS$.
  The second term is nonnegative, for otherwise $\MarMatchOpt$ could be augmented along $\Path$ to yield another matching $\MarMatchp$ such that $\PLeast{\Conf}{\MarMatchp}$ and $\WMarM{\Conf}{\MarMatchp} - \WMarM{\Conf}{\MarMatchOpt} \geq \NN$ (since all $\WomenPref$ values are integers), contradicting the definition of $\MarMatchOpt$.
  Thus, we conclude that $\WomanUtil{\Woman} = \WomanRes{\Woman}$.

  Case 2: The other endpoint of $\Path$ is a man $\Man$ such that $\IsMatchedS{\Man}$ and $\IsUnmatched{\Man}{\MarMatchOpt}$.
  This case is identical to case 2 in the proof of Claim~\ClaimMaxWeightUnmatchedMen\ above, and hence we conclude that $\WomanUtil{\Woman} = \WomanRes{\Woman}$.

  We now prove Claim~\ClaimMaxWeightUEqual\ of the lemma.
  Stability of $\MarOutcomeS$ implies $\ManUtil{\Man} = \ManResp{\Man}$ if $\IsUnmatchedS{\Man}$.
  Thus Claim~\ClaimMaxWeightUEqual\ follows using Claim~\ClaimMaxWeightUnmatchedMen\ and Lemma~\ref{lem:u-equal}.
  \QEDWrap
\end{proof}

\begin{proof}[\ConfOrFull{}{Proof of }Lemma~\ref{lem:max-weight-stable}]
  Let $\MarMatchOpt$ be a matching in $\MarMatchSet$ such that $\WMarM{\Conf}{\MarMatchOpt} = \max_{\MarMatchp \in \MarMatchSet} \WMarM{\Conf}{\MarMatchp}$.
  We show that $\WMarM{\Conf}{\MarMatch} = \WMarM{\Conf}{\MarMatchOpt}$ and that $\MarMatchOpt$ is compatible with the stable payoff $\MarPayoffS$.
  We have
  \begin{align}
    \WMarM{\Conf}{\MarMatchOpt} \geq \WMarM{\Conf}{\MarMatch}
    &= \sum_{\IsMatchedS{\Man}} \ManCompMatchS{\MarMatch} + \sum_{\EachWoman} \WomanUtilS + \MenPriUnmatchedS{\MarMatch} \notag\\
    &= \sum_{\IsMatched{\Man}{\MarMatchOpt}} \ManCompMatchS{\MarMatchOpt} + \sum_{\EachWoman} \WomanUtilS + \MenPriUnmatchedS{\MarMatchOpt} \label{eq:mar-match-weight-diff},
  \end{align}
  where the first equality follows from Lemma~\ref{lem:stable-weight}, and the second equality follows from Claim~\ClaimMaxWeightUEqual\ of Lemma~\ref{lem:max-weight-compatible}.
  Then, by using Claim~\ClaimMaxWeightUnmatchedWomen\ of Lemma~\ref{lem:max-weight-compatible} and Lemma~\ref{lem:women-payment}, we may rewrite~\eqref{eq:mar-match-weight-diff} as
  \begin{equation}
    \label{eq:mar-match-opt}
    0 \geq \sum_{\IsMatched{\Man}{\MarMatchOpt}} \ManCompMatchS{\MarMatchOpt} + \sum_{\IsMatched{\Woman}{\MarMatchOpt}} \WomanCompMatchS{\MarMatchOpt} = \sum_{\MOf{\MarMatchOpt}{\Man} = \Woman \land \Woman \not= \Unmatched} \sbr{\ManCompS + \WomanCompS} .
  \end{equation}
  Then, since the stability of $\MarPayoffS$ implies that $\ManCompS + \WomanCompS \geq 0$ for each man-woman pair $(\Man, \Woman)$, we deduce the following:
  the inequality of~\eqref{eq:mar-match-opt} is tight, and hence that of~\eqref{eq:mar-match-weight-diff} is also tight;
  $\ManCompS + \WomanCompS = 0$ for each man-woman pair $(\Man, \Woman)$ matched in $\MarMatchOpt$.
  Thus, Claims~\ClaimMaxWeightUnmatchedMen\ and \ClaimMaxWeightUnmatchedWomen\ of Lemma~\ref{lem:max-weight-compatible} imply that $\MarMatchOpt$ is compatible with the stable payoff $\MarPayoffS$.
  \QEDWrap
\end{proof}

\section{College Admissions with Indifferences} \label{app:college}

Our results for stable marriage markets with indifferences
can be extended to college admissions markets with indifferences
by transforming each student to a man and each slot of a college
to a woman, in a standard fashion.
For the sake of completeness, below
we provide the formal definition of the model
for college admissions markets with indifferences
and summarize our results.

The college admissions market involves a set $I$ of students
and a set $J$ of colleges. We assume that the sets
$I$ and $J$ are disjoint and do not contain the element $\nil$,
which we use to denote being unmatched.
The preference relation of each student $i \in I$
is specified by a binary relation $\succeq_i$ over $J \cup \{\jnil\}$
that satisfies transitivity and totality.
Similarly, the preference relation of each college $j \in J$
over individual students is specified by a binary relation
$\succeq_j$ over $I \cup \{\inil\}$
that satisfies transitivity and totality.
Furthermore, each college $j \in J$ is associated with a capacity $c_j$. 
We denote this \emph{college admissions market} as $(I, J,
(\succeq_i)_{i \in I}, (\succeq_j)_{j \in J}, c)$.

The colleges' preference relations over individual students can be
extended to group preference relations using
the notion of responsiveness
introduced by Roth~\cite{Rot85}. We say that
a transitive and reflexive relation $\succeq_j'$ over the power
set $2^I$ is \emph{responsive to the preference relation $\succeq_j$}
if the following conditions hold.
\begin{enumerate}
\item For any $I_0 \subseteq I$ and $i_1 \in I \setminus I_0$,
we have $i_1 \succeq_j \inil$ if and only if
$I_0 \cup \{i_1\} \succeq_j' I_0$.
\item For any $I_0 \subseteq I$ and $i_1, i_2 \in I \setminus I_0$,
we have $i_1 \succeq_j i_2$ if and only if
$I_0 \cup \{i_1\} \succeq_j' I_0 \cup \{i_2\}$.
\end{enumerate}
Furthermore, we say that a relation $\succeq_j'$ is
\emph{minimally responsive\footnote{
Our notion of minimal responsiveness is the same as that
studied by Erdil and Ergin~\cite{EE08,EE15},
Chen~\cite{Che12}, Chen and Ghosh~\cite{CG10},
and Kamiyama~\cite{Kam14}, though they did not use this terminology.}
to the preference relation $\succeq_j$}
if it is responsive to the preference relation $\succeq_j$
and does not strictly contain
another relation that
is responsive to the preference relation $\succeq_j$.

A \emph{(capacitated) matching} is a function
$\mu \colon I \to J \cup \{\jnil\}$ such that
for any college $j \in J$, we have $\mu(i) = j$ for at most $c_j$
students $i \in I$. Given a matching $\mu$ and a college $j \in J$,
we denote $\mu(j) = \{ i \in I \colon \mu(i) = j \}$.

A matching $\mu$ is \emph{individually rational} if
$j \succeq_i \jnil$ and $i \succeq_j \inil$
for every student $i \in I$ and college $j \in J$ such that $\mu(i) = j$.
An individually rational matching $\mu$ is \emph{weakly stable} if
for any student $i \in I$ and college $j \in J$, either
$\mu(i) \succeq_i j$ or both of the following conditions hold.
\begin{enumerate}
\item For every student $i' \in I$ such that $\mu(i') = j$,
  we have $i' \succeq_j i$.
\item Either $\abs{\mu(j)} = c_j$ or $\inil \succeq_j i$.
\end{enumerate}
(Otherwise, such a student $i$ and college $j$ form a \emph{strongly blocking pair}.)

Let $(\succeq_j')_{j \in J}$ be the group preferences associated with
the colleges. For any matchings $\mu$ and $\mu'$, we say that the
binary relation $\mu \succeq \mu'$ holds if $\mu(i) \succeq_i \mu'(i)$
and $\mu(j) \succeq_j' \mu'(j)$ for every student $i \in I$
and college $j \in J$. A weakly stable matching $\mu$
is \emph{Pareto-stable} if for every matching $\mu'$ such that
$\mu' \succeq \mu$, we have $\mu \succeq \mu'$.
(Otherwise, the matching $\mu$ is
not \emph{Pareto-optimal} because it is \emph{Pareto-dominated}
by the matching $\mu'$.)

A \emph{mechanism} is an algorithm that, given a college admissions
market $(I, J, \allowbreak {(\succeq_i)_{i \in I}}, \allowbreak
{(\succeq_j)_{j \in J}} , c)$, produces a matching $\mu$.
A mechanism is said to be 
\emph{group strategyproof (for the students)}
if for any two different preference profiles
$(\succeq_i)_{i \in I}$ and $(\succeq'_i)_{i \in I}$,
there exists a student $i_0 \in I$ with preference relation
$\succeq_{i_0}$ different from $\succeq_{i_0}'$ such that
$\mu(i_0) \succeq_{i_0} \mu'(i_0)$, where $\mu$ and $\mu'$
are the matchings produced by the mechanism given $(I, J,
(\succeq_i)_{i \in I}, (\succeq_j)_{j \in J}, c)$ and 
$(I, J, (\succeq'_i)_{i \in I}, (\succeq_j)_{j \in J}, c)$,
respectively. (Such a student $i_0$ belongs to the coalition but
is not matched to a strictly preferred college
by expressing preference relation $\succeq_{i_0}'$
instead of his true preference relation $\succeq_{i_0}$.)

Given a college admissions market
$(I, J, (\succeq_i)_{i \in I}, (\succeq_j)_{j \in J}, c)$,
we can construct a tiered-slope market
$\calM = (I, J^*, \pi, N, \lambda, a, b)$ with
the set of men being the set $I$ of students
and the set of women being the set
\begin{equation*}
J^* = \{ (j, k) \in J \times \bbZ \colon 1 \leq k \leq c_j \}
\end{equation*}
of available slots in the colleges in a similar way
as described in Sect.~\ref{sec:marriage}.
If $j^*_1 = (j, k_1)$ and $j^*_2 = (j, k_2)$ are slots
in the same college $j$, then we simply have
$b_{i, j^*_1} = b_{i, j^*_2}$
for every $i \in I \cup \{\inil\}$ and
$a_{i, j^*_1} = a_{i, j^*_2}$ for every $i \in I$.
So, if college $j$ weakly prefers a subset $I_1$ of students
to another subset $I_2$ of students under its minimally
responsive preferences, then the total utility of
college $j$ for the students in $I_1$ plus the utilities
for its unmatched slots is greater than or equal
to that of $I_2$.

Using this tiered-slope market $\calM$, we can compute
a man-optimal outcome $(\mu^*, u, v)$. The matching $\mu^*$
from students in $I$ to slots in $J^* \cup \{\jnil\}$
induces a matching $\mu$ from students in $I$ to
colleges in $J \cup \{\jnil\}$.
Pareto-stability, group strategyproofness, and
polynomial-time computability all generalize
to the college admissions setting
in a straightforward manner.
The theorem below summarizes these results.

\begin{theorem}
There exists a polynomial-time algorithm that corresponds
to a group strategyproof and Pareto-stable
mechanism for the college admissions market with indifferences,
assuming that the group preferences of the colleges
are minimally responsive.
\end{theorem}

\end{document}